\documentclass[10pt, conference]{IEEEtran}
\usepackage{epsfig,endnotes}
\usepackage{multirow}

\usepackage[utf8]{inputenc}
\usepackage[T1]{fontenc}
\usepackage[bookmarks=false]{hyperref}
\usepackage{balance}
\hypersetup{colorlinks,%
            citecolor=black,%
            filecolor=black,%
            linkcolor=black,%
            urlcolor=black,%
            pdftex}
\usepackage{url}
\usepackage{amsfonts}
\usepackage{amssymb}
\usepackage{amsthm}
\usepackage{booktabs}
\usepackage{multirow}
\usepackage[vlined, ruled, commentsnumbered, linesnumbered]{algorithm2e}
\usepackage{booktabs}
\usepackage{paralist}
\usepackage{comment}
\usepackage{caption}
\usepackage{subcaption}
\usepackage{datetime}
\usepackage[all]{nowidow}
\usepackage[hang,flushmargin]{footmisc}

\usepackage{blindtext}
\usepackage{graphicx}
\usepackage{tcolorbox}

\definecolor{mygray}{gray}{0.5}

\usepackage{xspace}
\usepackage{xcolor}

\newtheorem{theorem}{Theorem}

\newtheorem{case}{Case}
\newtheorem{case2}{Case}
\newtheorem{case3}{Case}

\definecolor{darkgreen}{rgb}{0,0.5,0}
\definecolor{brown}{rgb}{0.7,0.3,0}
\definecolor{darkblue}{rgb}{0,0,0.5}
\definecolor{darkred}{rgb}{0.5,0,0}
\definecolor{mygray}{gray}{0.5}

\newcommand{\aviv}[1]{\textit{\textbf{\textcolor{brown}{Aviv: #1}}}}
\newcommand{\laurent}[1]{\textit{\textbf{\textcolor{red}{Laurent: #1}}}}
\newcommand{\maria}[1]{\textit{\textbf{\textcolor{darkgreen}{Maria: #1}}}}

\newcommand{\myitem}[1]{\vspace*{0.07in}\noindent\textbf{#1}}

\newcommand{\remove}[1]{}

\newcounter{reviewcounter}

\usepackage{enumitem}
\usepackage{tikz} 
\usetikzlibrary{backgrounds,calc,arrows,decorations.pathmorphing,backgrounds,fit,positioning,shapes.symbols,chains}

\newif\ifHideComments

\ifHideComments
	\renewcommand{\aviv}[1]{}
	\renewcommand{\laurent}[1]{}
	\renewcommand{\maria}[1]{}

	\excludecomment{notes}
	\excludecomment{thesiscontents}
\fi

\begin{document}

\title{\huge Hijacking Bitcoin: Routing Attacks on Cryptocurrencies \\[.35em] \small \texttt{\url{https://btc-hijack.ethz.ch}}\vspace{-.35em}}
\author{\IEEEauthorblockN{Maria Apostolaki}
\IEEEauthorblockA{ETH Z\"urich\\
apmaria@ethz.ch}
\and
\IEEEauthorblockN{Aviv Zohar}
\IEEEauthorblockA{The Hebrew University\\
avivz@cs.huji.ac.il}
\and
\IEEEauthorblockN{Laurent Vanbever}
\IEEEauthorblockA{ETH Z\"urich\\
lvanbever@ethz.ch}}


\maketitle

\begin{abstract}
As the most successful cryptocurrency to date, Bitcoin constitutes a target of choice for attackers. While many attack vectors have already been uncovered, one important vector has been left out though: attacking the currency via the
Internet routing infrastructure itself. Indeed, by manipulating routing
advertisements (BGP hijacks) or by naturally intercepting traffic,
Autonomous Systems (ASes) can intercept and manipulate a large fraction of Bitcoin
traffic. 

This paper presents the first taxonomy of routing attacks and their impact on
Bitcoin, considering both small-scale attacks, targeting individual nodes, and
large-scale attacks, targeting the network as a whole. While challenging, we
show that two key properties make routing attacks practical: \emph{(i)} the
efficiency of routing manipulation; and \emph{(ii)} the significant
centralization of Bitcoin in terms of mining and routing. Specifically, we
find that any network attacker can hijack few (<100) BGP prefixes to isolate
{\raise.17ex\hbox{$\scriptstyle\sim$}}50\% of the mining power---even when
considering that mining pools are heavily multi-homed. We also show that
on-path network attackers can considerably slow down block propagation by interfering with few key Bitcoin messages.

We demonstrate the feasibility of each attack against the deployed Bitcoin
software. We also quantify their effectiveness on the current Bitcoin topology
using data collected from a Bitcoin supernode combined with BGP routing data.

The potential damage to Bitcoin is worrying. By isolating parts of the network
or delaying block propagation, attackers can cause a significant amount of
mining power to be wasted, leading to revenue losses and enabling a wide range
of exploits such as double spending. To prevent such effects in practice, we provide both short and long-term countermeasures, some of which can be
deployed immediately.

\end{abstract}

\remove{We demonstrate the feasibility of each attack against the deployed Bitcoin software and quantify both the node-level and network-wide impact using data collected from a bitcoin supernode combined with BGP routing data. The potential damage to Bitcoin is worrying. }
\section{Introduction}
\label{sec:introduction}

With more than 16 million bitcoins valued at {\raise.17ex\hbox{$\scriptstyle\sim$}}17 billion USD and up to
300{,}000 daily transactions (March 2017), Bitcoin is the most successful
cryptocurrency to date. Remarkably, Bitcoin has achieved this as an open and fully
decentralized system. Instead of relying on a central entity, Bitcoin nodes build a large overlay network between them and use consensus to
agree on a set of transactions recorded within Bitcoin's core data structure: the
blockchain. Anyone is free to participate in the network which boasts
more than 6{,}000 nodes~\cite{bitnodes21} and can usually connect to
any other node. 

Given the amount of money at stake, Bitcoin is an obvious target
for attackers. Indeed, numerous attacks have been described targeting different aspects of
the system including: double spending~\cite{rosenfeld2014analysis}, eclipsing~\cite{heilman2015eclipse}, transaction malleability~\cite{decker2014bitcoin}, or attacks targeting mining~\cite{eyal2014majority,DBLP:journals/corr/SapirshteinSZ15,DBLP:journals/iacr/NayakKMS15} and mining pools~\cite{eyal2015miner}. 

One important attack vector has been overlooked though: attacking
Bitcoin via the Internet infrastructure using \emph{routing attacks}. As
Bitcoin connections are routed over the Internet---in clear text and without
integrity checks---any third-party on the forwarding path can eavesdrop, drop, modify, inject, or delay
Bitcoin messages such as blocks or transactions. Detecting such attackers is
challenging as it requires inferring the exact forwarding paths taken by the Bitcoin traffic using measurements (e.g., traceroute) or routing data (BGP
announcements), both of which can be forged~\cite{defcon_mitm}. Even ignoring
detectability, mitigating network attacks is also hard as it is essentially a
human-driven process consisting of filtering, routing around or disconnecting
the attacker. As an illustration, it took Youtube close to 3 hours to locate
and resolve rogue BGP announcements targeting its infrastructure in
2008~\cite{dyn_pakistani_bgp}. More recent examples of routing attacks such as
\cite{hijack_november_15} (resp. \cite{hijack_june_15}) took 9 (resp. 2) hours
to resolve in November (resp. June) 2015.


One of the reasons why routing attacks have been overlooked in Bitcoin is that they
are often considered too challenging to be practical. Indeed, perturbing a
vast peer-to-peer network which uses random flooding is hard as an attacker
would have to intercept many connections to have any impact. 
Yet, two key characteristics of the Internet's infrastructure make routing
attacks against Bitcoin possible: \emph{(i)} the
efficiency of routing manipulation (BGP hijacks); and \emph{(ii)} the centralization 
of Bitcoin from the routing perspective.
First, individuals, located anywhere on the Internet, can manipulate
routing to intercept all the connections to not only
one, but many Bitcoin nodes. As we show in this paper, these routing manipulations are
prevalent today and do divert Bitcoin traffic. Second, few ASes host
most of the nodes and mining power, while others intercept a considerable fraction of the connections.

\myitem{This work} In this paper, we present the first taxonomy of routing
attacks on Bitcoin, a comprehensive study of their impact, and a list of deployable countermeasures. We consider two
general attacks that AS-level attackers can perform.
First, we evaluate the ability of attackers to isolate a set of nodes from the Bitcoin network, effectively partitioning it. Second, we evaluate the impact of delaying block
propagation by manipulating a small number of key Bitcoin messages. For both
exploits, we consider \emph{node-level} attacks along with more challenging,
but also more disruptive, \emph{network-wide} attacks.

\noindent\textbf{Partitioning attacks} The goal of a partition attack is to \emph{completely} disconnect a set of nodes from the network. This requires the attacker to divert and cut all the connections between the set of nodes and the rest of the network. 

We describe a complete attack procedure in which an attacker can
\emph{verifiably} isolate a selected set of nodes using BGP hijacks. Our
procedure is practical and only requires basic knowledge of the Bitcoin
topology, namely the IP addresses of the nodes the attacker wants to isolate.
Due to the complexity of the Bitcoin network (e.g. multi-homed pools, and
secret peering agreements between pools), the initial isolated set might contain
nodes that leak information from and to the rest of the network. We explain how
the attacker can identify and remove these leakage points until the partition
is complete.

\myitem{Delay attacks} The goal of a delay attack is to slow down the
propagation of blocks towards or from a given set of nodes. Unlike partition
attacks, which require a perfect cut, delay attacks
are effective even when a subset of the connections are intercepted. As
such, attackers can perform delay attacks on connections they are naturally intercepting, making them even harder to detect.

We again describe a complete attack procedure an attacker can run on
intercepted Bitcoin traffic so that the delivery of blocks
is delayed by up to 20 minutes. The procedure consists of modifying
 few key Bitcoin messages while making sure that the connections are not
disrupted. 

\remove{
\myitem{Challenges} Large mining pools play a
critical role in the Bitcoin network. We therefore investigate the influence of
their connectivity (i.e., how much they are multihomed) on the effectiveness of
routing attacks. We show that isolation attacks do little dammage if the
partition is imperfect (e.g., due to non-intercepted gateways connections).
However, we describe ways in which strategic AS-level adversaries can actually
discover hidden gateways preventing the partition to be formed and hijack them.
Regarding delay attacks, we show that network-wide attacks do indeed become too
difficult for attackers (if powerful ones like country-wide coalitions) if
pools are multi-homed.
}

\myitem{Practicality} We showcase the practicality of each attack and evaluate
their network-wide impact using a comprehensive set of measurements,
simulations and experiments.

Regarding partitioning attacks, we show that hijacks are effective in diverting
Bitcoin traffic by performing a hijack in the wild against our own nodes. We
find that it takes less than 90 seconds to re-route all traffic flows through
the attacker once a hijack is initiated. We also show that \emph{any AS} in the
Internet hijacking \emph{less than 100} prefixes can isolate up to 47\% of the mining power, and this,
even when considering that mining pools are multi-homed. Hijacks involving
that many prefixes are frequent and already divert Bitcoin traffic.

Regarding delay attacks, we show that an attacker intercepting 50\% of a
node connections can leave it uninformed of the most recent Bitcoin blocks
{\raise.17ex\hbox{$\scriptstyle\sim$}}60\% of the time. 
We also show that intercepting a considerable percentage of Bitcoin
traffic is practical due to the centralization of Bitcoin at the routing level: one AS, namely Hurricane Electric, can \emph{naturally} 
intercept more than 30\% of \emph{all} Bitcoin connections.

\remove{
We also show that intercepting large amounts of Bitcoin
traffic is practical due to the important centralization of Bitcoin at the
the routing and mining level: 13 ASes host over 30\%
of the Bitcoin network. 
Yet, we also show that network-wide delay attacks, aiming at slowing down the
entire Bitcoin network, are not practical even for powerful
attackers.
}


\myitem{Impact on Bitcoin} The damages caused to Bitcoin in case of a
successful routing attack can be substantial. By isolating a part of the
network or delaying the propagation of blocks, attackers can force nodes to
waste part of their mining power as some of the blocks they create are
discarded. Partitioning also enables the attacker to filter transactions that
clients try to include in the blockchain. In both cases, miners lose potential
revenue from mining and render the network more susceptible to double spending
attacks as well as to selfish mining attacks~\cite{eyal2014majority}. Nodes representing merchants, exchanges and other large entities are thus unable to secure their transactions, or may not be able to broadcast
them to the network to begin with. The resulting longer-term loss of trust in
Bitcoin security may trigger a loss of value for Bitcoin. Attackers may even
short Bitcoin and gain from the resulting devaluation~\cite{kroll2013economics}.

Our work underscores the importance of proposed modifications which
argue for encrypting Bitcoin traffic~\cite{bitcoin_bip_151} or traffic exchanged among miners~\cite{lef}. Yet, we
stress that not all routing attacks will be solved by such
measures since attackers can still disrupt connectivity and isolate nodes by dropping Bitcoin packets instead of modifying them.

\myitem{Contributions} Our main contributions are:\footnote{Our software, measurements and scripts can be found online at \\ \texttt{\url{https://btc-hijack.ethz.ch}}}

\begin{itemize}[leftmargin=*]
\setlength{\itemsep}{1pt}
\item The first comprehensive study of network attacks on Bitcoin
(Section~\ref{sec:overview}) ranging from attacks targeting a single node to attacks affecting the network as a whole.
\item A measurement study of the routing properties of Bitcoin (Section~\ref{sec:topology}). We show that Bitcoin is highly centralized: few ASes host most of the nodes while others intercept a considerable fraction of the connections. \remove{({\raise.17ex\hbox{$\scriptstyle\sim$}}30\%)}
\item A thorough evaluation of the practicality of routing attacks
(partitioning and delay attacks). Our evaluation is based on an extensive set of
measurements, large-scale simulations and experiments on the actual Bitcoin
software and network.
\remove{
\item An in-depth evaluation on persistence  of partition attacks
using a large testbed of 1{,}050 Bitcoin nodes (Section~\ref{sec:eval_partition}).
We show that partitions are not persistent after the attack due to the natural churn of the network.
}
\remove{
\item A fully-working implementation of a Bitcoin interception software
(Section~\ref{sec:delay_results}) that AS-level adversaries can use to delay block propagation; along with the implementation of a realistic
event-driven Bitcoin simulator to measure the network-wide effect of delay
attacks. We show that delay attacks are very effective for small subset of nodes, but not realistic at the network-wide level if pools are multi-homed enough.
}
\item A comprehensive set of countermeasures (Section~\ref{sec:countermeasures}), which can benefit even early adopters.

\end{itemize}

While our measurements are Bitcoin-specific, they carry important
lessons for other cryptocurrencies which rely on a
randomly structured peer-to-peer network atop of the Internet, such as Ethereum~\cite{ethereum}, Litecoin~\cite{litcoin}, and ZCash~\cite{zcash,sasson2014zerocash}.

\section{Background}
\label{sec:background}

\subsection{BGP}

\myitem{Protocol} BGP~\cite{rfc1771} is the de-facto routing protocol that
regulates how IP packets are forwarded in the Internet. Routes associated with
different IP prefixes are exchanged between neighboring networks or Autonomous
Systems (AS). For any given IP prefix, one AS (the origin) is responsible for
the original route advertisement, which is then propagated AS-by-AS until all
ASes learn about it. Routers then set their next hop and pick one of the
available routes offered by their neighbors (this is done independently for
each destination).

In BGP, the validity of route announcements is not checked. In effect, this
means that any AS can inject forged information on how to reach one or
more IP prefixes, leading other ASes to send traffic to the wrong location.
These rogue advertisements, known as BGP ``hijacks'', are a very effective
way for an attacker to intercept traffic en route to a legitimate
destination.

\myitem{BGP hijack} An attacker, who wishes to attract all the traffic for a legitimate prefix $p$ (say,
100.0.0.0/16) by hijacking could either: \emph{(i)} announce $p$; or \emph{(ii)} announce a more-specific (longer) prefix of $p$. In the first
case, the attacker's route will be in direct competition with the legitimate
route. As BGP routers prefer shorter paths, the attacker will, on average,
attract 50\% of the traffic~\cite{goldberg_how_secure_are_interdomain_routing_protocols}. In the
second case, the attacker will attract all the traffic (originated anywhere on the Internet) addressed to the destination as Internet routers forward traffic according to the longest-match entry. 
Note that traffic internal to an AS cannot be diverted via hijacking as it does not get routed by BGP but by internal routing protocols (e.g., OSPF).

For instance, in order to attract all traffic destined to
$p$, the attacker could advertise 100.0.0.0/17 \emph{and} 100.0.128.0/17.
Routers in the entire Internet would then start forwarding any traffic destined to the original /16 prefix according to the two covering /17s originated
by the adversary. Advertising more-specific prefixes has its limits though as BGP operators will often
filter prefixes longer than /24~\cite{Karlin:2006:PGB:1317535.1318378}. Yet, we show
that the vast majority of Bitcoin nodes is hosted in shorter prefixes
(Section~\ref{sec:topology}) and is thus susceptible to hijacking.

By default, hijacking a prefix creates a black hole at the attacker's location. However, the attacker can turn a hijack into an
\emph{interception} attack simply by making sure she leaves at least one path untouched to the
destination~\cite{defcon_mitm, goldberg_how_secure_are_interdomain_routing_protocols}.

\subsection{Bitcoin}


\myitem{Transactions} Transaction validation requires nodes to be aware of the
ownership of funds and the balance of each Bitcoin address. All this
information can be learned from the Bitcoin blockchain: an authenticated data
structure that effectively forms a ledger of all accepted transactions.
Bitcoin main innovation lies in its ability to synchronize the blockchain in
an asynchronous way, with attackers possibly attempting to disrupt the
process. Synchronization is crucial: conflicting transactions attempting to
transfer the exact same bitcoins to different destinations may otherwise be
approved by miners that are unaware of each other.


\myitem{Block creation} Bitcoin's blockchain is comprised of blocks, batches of transactions, that are
appended to the ledger serially. Each block contains a cryptographic hash of
its predecessor, which identifies its place in the chain, and a proof-of-work.
The proof-of-work serves to make block creation difficult and reduces the
conflicts in the system. Conflicts, which take the form of blocks that extend
the same parent, represent alternative sets of accepted transactions. Nodes
converge to a single agreed version by selecting the chain containing the
highest amount of computational work as the valid version (usually the longest
chain). The proof-of-work also serves to limit the ability of
attackers to subvert the system: they cannot easily create many blocks, which
would potentially allow them to create a longer alternative chain that will
be adopted by nodes and thus reverse the transfer of funds (double spend).

The difficulty of block creation is set so that one block is created in the
network every 10 minutes on average which is designed to allow
sufficient time for blocks to propagate through the network. However, if delays
are high compared to the block creation rate, many forks occur in the chain as
blocks are created in parallel. In this case, the rate of discarded blocks (known as the
\emph{orphan rate} or the \emph{fork rate}) increases and the security of the
protocol deteriorates~\cite{decker2013information,garay2015bitcoin,sompolinsky2015secure}. Newly
created blocks are propagated through the network using a gossip protocol. In
addition to the propagation of blocks, nodes also propagate transactions
between them that await inclusion in the chain by whichever node creates
the next block.

\myitem{Network formation} Bitcoin acts as a
peer-to-peer network with each node maintaining a list of IP addresses of
potential peers. The list is bootstrapped via a DNS server, and additional
addresses are exchanged between peers. By default, each node randomly initiates 8 \emph{unencrypted} TCP connections to peers in
different /16 prefixes. Nodes additionally accept connections initiated by
others (by default on port 8333). The total number of connections nodes
can make is 125 by default.

Nodes continually listen to block announcements which are sent via \textsf{INV} messages containing the hash of the
announced block. If a node determines that it does not hold a newly announced block, it sends a \textsf{GETDATA} message to a \emph{single}
neighbor. The peer then responds by sending the requested information in a
\textsf{BLOCK} message. Blocks that are requested and do not arrive within 20
minutes trigger the disconnection of the peer and are requested from another.
Transaction propagation occurs with a similar sequence of \textsf{INV},
\textsf{GETDATA}, and \textsf{TX} messages in which nodes announce, request,
and share transactions that have not yet been included in the blockchain.

\myitem{Mining pools}
Mining pools represent groups of miners that divide block creation rewards between them in order to lower the high
economic risk associated with infrequent (but high) payments. They usually operate using the Stratum protocol \cite{bitcoin:book}. The
pool server is connected to a bitcoind node that acts as a gateway to the Bitcoin network. The node collects
recent information regarding newly transmitted transactions and newly built blocks which are then used to construct a
new block template. The template header is then sent via the Stratum server to the miners who attempt to complete
it to a valid block. This is done by trying different values of the nonce field in the header. 
If the block is completed, the result is sent back to the Stratum server,
which then uses the gateway node to publish the newly formed block to the network.

\myitem{Multi-homing} Mining pools often use multiple gateways hosted by diffrent Internet
Service Providers. We refer to the number of different ISPs a
pool has as its multi-homing degree. 

\section{Routing attacks on Bitcoin}
\label{sec:overview}

\begin{figure}[t]
 \centering
 \includegraphics[width=\columnwidth]{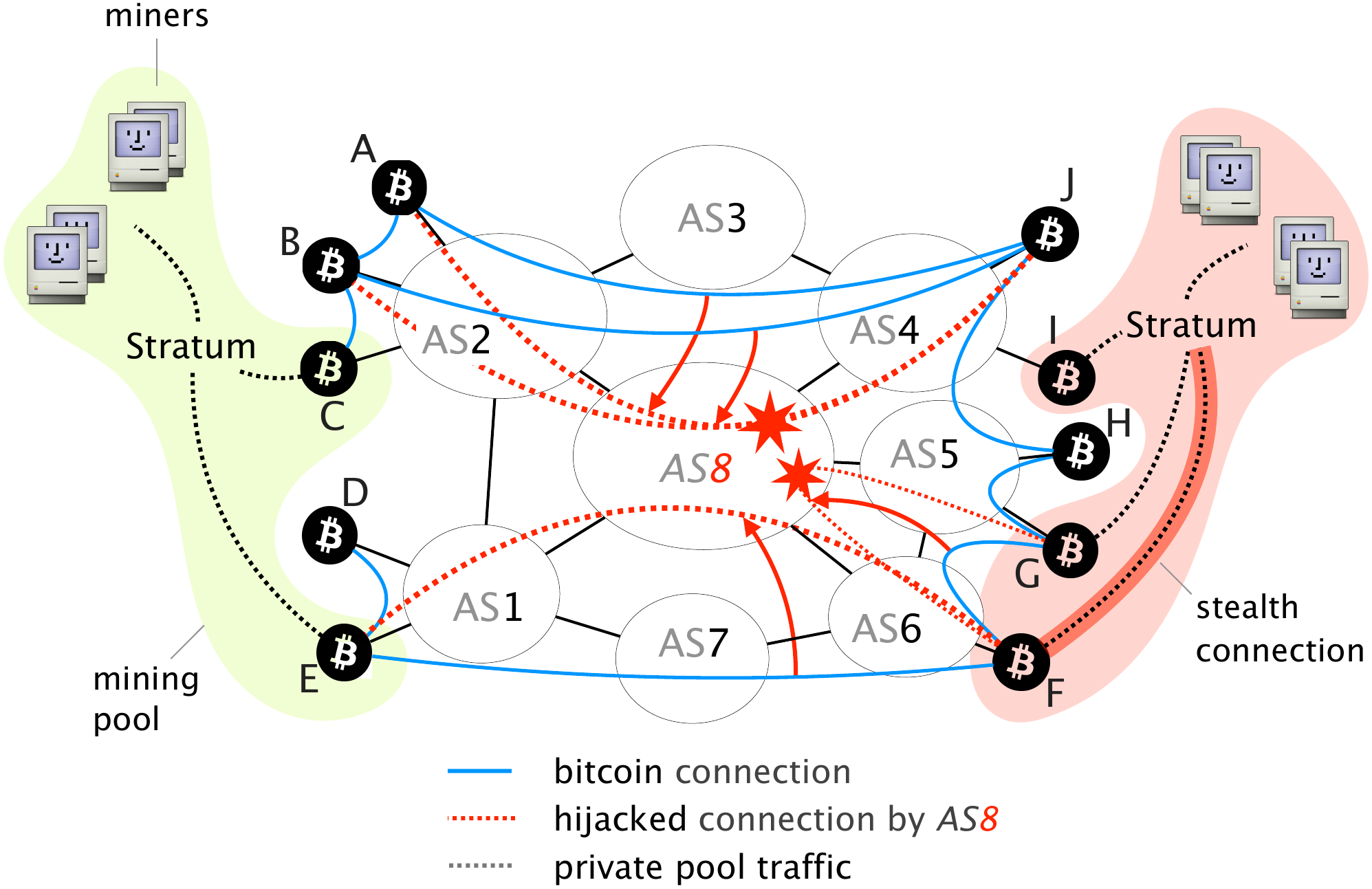}
 \caption{Illustration of how an AS-level adversary (AS8) can intercept Bitcoin traffic by hijacking prefixes to isolate the set of nodes $P = (A,
B, C, D, E, F)$.}
 \label{fig:overview_partition}
\end{figure}

In this section, we give an overview of the two routing attacks we
describe in this paper: \emph{(i)} partitioning the Bitcoin network (Section~\ref{ssec:partition_overview}); and
\emph{(ii)} delaying the propagation of blocks. For each attack, we briefly
describe its effectiveness and challenges as well as its impact on the Bitcoin
ecosystem (Section~\ref{ssec:delay_overview}).

\remove{
\begin{figure}[t]
 \centering
 \includegraphics[width=\columnwidth]{figures/partition}
 \caption{Illustration of how an AS-level adversary (AS8) can intercept Bitcoin traffic by hijacking prefixes or leveraging natural connectivity and then attacking either specific nodes of the network as a whole. \laurent{Update caption}}
 \label{fig:partition_sumary}
\end{figure}
}

\subsection{Partitioning the Bitcoin Network}
\label{ssec:partition_overview}

In this attack, an AS-level adversary seeks to isolate a set of nodes $P$ from
the rest of the network, effectively partitioning the Bitcoin network into two
disjoint components. The actual content of $P$ depends on the attacker's
objectives and can range from one or few merchant nodes, to a set of nodes holding a considerable percentage of the total mining power.

\myitem{Attack} The attacker first diverts the traffic destined to nodes in
$P$ by hijacking the most-specific prefixes hosting each of the IP address. Once on-path, the attacker intercepts the Bitcoin traffic
(e.g., based on the TCP ports) and identifies whether the
corresponding connections cross the partition she tries to create. If so, the attacker drops the packets. If not, meaning the connections are contained within $P$, she monitors the exchanged Bitcoin messages so as to detect ``leakage points''. Leakage points are nodes currently within $P$, which maintain connections with nodes outside of $P$, that the attacker cannot intercept, namely ``stealth'' connections. The attacker can detect these nodes automatically and isolate them from others in $P$ (Section~\ref{sec:partition}). Eventually, the attacker isolates the maximal set of nodes in $P$ that can be isolated.

\myitem{Example} We illustrate the partition attack on the simple network in Fig.~\ref{fig:overview_partition} that is composed of 8 ASes, some of which host Bitcoin nodes. Two mining
pools are depicted as a green (left) and a red (right) region. Both pools
are multi-homed and have gateways in different ASes. For instance, the
red (right) pool has gateways hosted in AS4, AS5, and AS6. We denote the
initial Bitcoin connections with blue lines, and those that have been diverted via hijacking with red lines. Dashed black lines
represent private connections within the pools. Any AS on the path of a connection can intercept it.

Consider an attack by AS8 that is meant to isolate the set of nodes $P = (A,
B, C, D, E, F)$. First, it hijacks the prefixes
advertised by AS1, AS2 and AS6, as they host nodes within $P$,
effectively attracting the traffic destined to them. Next, AS8 drops all
connections crossing the partition: i.e., $(A,J)$, $(B,J)$ and $(F,G)$. 

Observe that node $F$ is within the isolated set $P$, but is also a gateway of the red pool with which $F$ most likely communicates. This connection may not be based on the Bitcoin protocol
and thus it cannot be intercepted (at least, not easily). As such, even if the attacker
drops all the Bitcoin connections she intercepts, node $F$ may still learn about transactions
and blocks produced on the other side and might leak this
information within $P$. Isolating $P$ as such is infeasible. However, AS8 can identify that node $F$ is the leakage point during the attack and exclude it
from $P$, essentially isolating $I'= (A, B, C, D, E)$ instead. This
$I'$ is actually the maximum subset of $P$ that can be isolated from the
Bitcoin network.


\myitem{Practicality} We extensively evaluate the practicality of isolating sets of nodes of various sizes
(Section~\ref{sec:eval_partition}). We briefly summarize our findings.
\emph{First}, we performed a real BGP hijack against our own Bitcoin nodes
and show that it takes less than 2 minutes for an attacker to divert Bitcoin
traffic. \emph{Second}, we estimated the number of prefixes to hijack so as to
isolate nodes with a given amount of mining power. We found that hijacking only 39
prefixes is enough to isolate a specific set of nodes which accounts for almost 50\% of the overall mining
power. Through a longitudinal analysis spanning over 6 months, we
found that much larger hijacks happen regularly and that some of them have
already impacted Bitcoin traffic. \emph{Third}, we show that, while effective, partitions do not last long after the attack stops: the two components of the partition quickly reconnect, owing to natural churn. Yet, it takes hours for the two components to be densely connected again.

\myitem{Impact} The impact of a partitioning attack depends on the number of
isolated nodes and how much mining power they have. Isolating a few nodes essentially constitutes a denial of service attack and renders them
vulnerable to 0-confirmation double spends. Disconnecting a considerable amount of mining power can lead to the creation of two
different versions of the blockchain. All blocks mined on the side with the
least mining power will be discarded and all included transactions are likely
to be reversed. Such an attack would cause revenue loss for the miners on the
side with least mining power and a prominent risk of double spends. The side
with the most mining power would also suffer from an increased risk of selfish
mining attacks by adversaries with mining power.

\begin{figure}[t]
 \centering
 \includegraphics[width=.75\columnwidth]{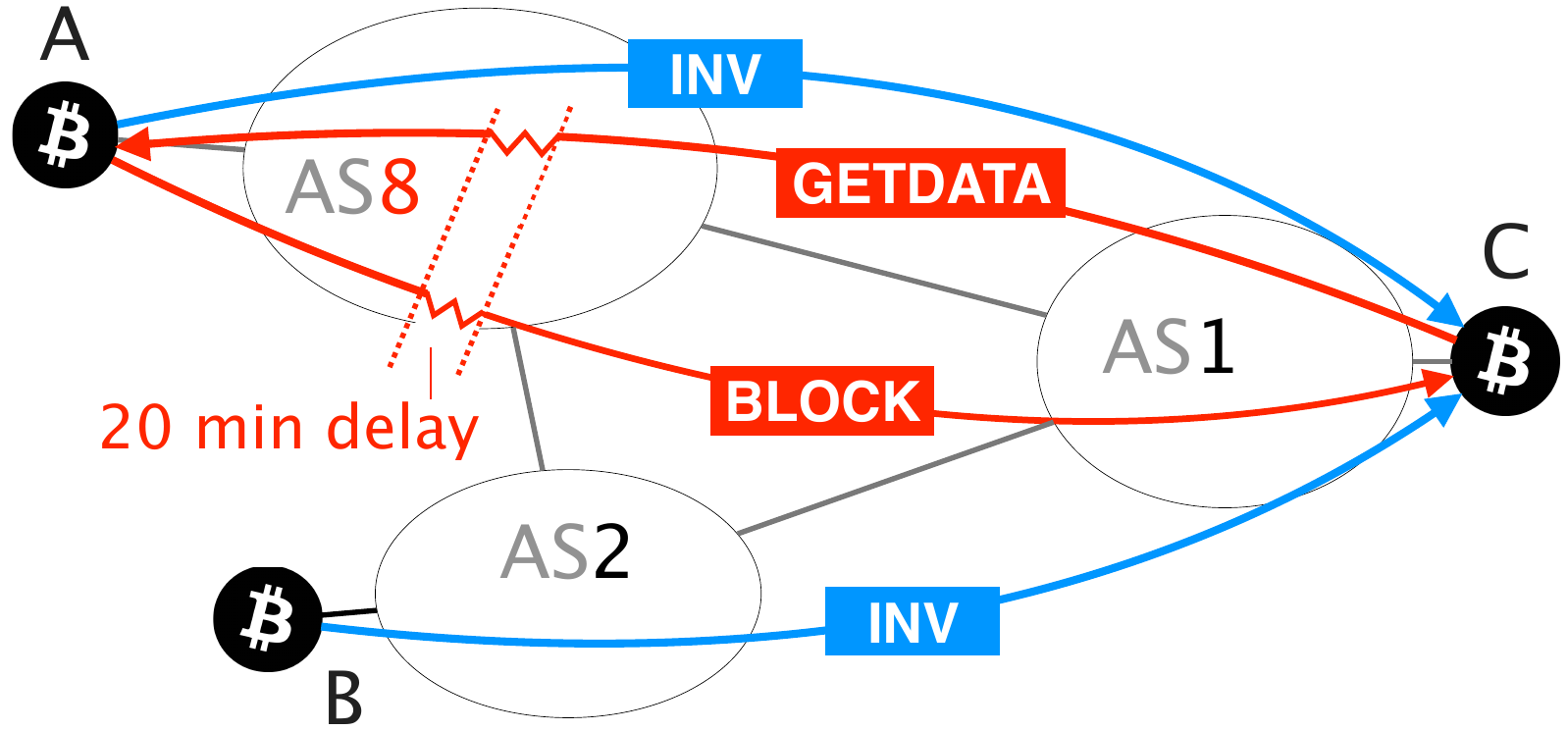}
 \caption{Illustration of how an AS-level adversary (AS8) which naturally intercepts a part of the traffic can delay the delivery of a block for 20 minutes to a victim node ($C$).}
 \label{fig:delay_overview}
\end{figure}

\subsection{Slowing down the Bitcoin network}
\label{ssec:delay_overview}

In a delay attack, the attacker's goal is to slow down the propagation of new
blocks sent to a set of Bitcoin nodes without disrupting their connections.
As with partitioning, the attack can be targeted, aimed at selected nodes,
or network-wide, aimed at disrupting the ability of the entire network to reach
consensus~\cite{decker2013information}. Unlike partitioning attacks though, an
attacker can delay the overall propagation of blocks towards a node even if
she intercepts a subset of its connections.

\myitem{Attack} Delay attacks leverage three key aspects of the Bitcoin
protocol: \emph{(i)} the asymmetry in the way Bitcoin nodes exchange blocks
using \textsf{INV}, \textsf{GETDATA}, and \textsf{BLOCK} messages
(Section~\ref{sec:background}); \emph{(ii)} the fact that these messages are
not protected against tampering (unencrypted, no secure integrity checks); and \emph{(iii)} the fact that a Bitcoin node
waits for 20 minutes after having requested a block from a peer before requesting
it again from another peer.
These protocol features enable an attacker intercepting even one direction of the victim's connection
to delay the propagation of a block, as long as this connection is traversed by either the actual
  \textsf{BLOCK} message or the corresponding \textsf{GETDATA}.

Specifically, if the attacker intercepts the traffic \emph{from} the victim,
she can modify the content of the \textsf{GETDATA} message the victim
uses to ask for blocks. By preserving the message length and structure and by
updating the TCP and Bitcoin checksums, the modified message is accepted by the
receiver and the connection stays alive. \remove{Next, the attacker can yet
again tamper with another \textsf{GETDATA} message to trigger a delayed (up to
20 minutes) delivery of the block to prevent the victim from disconnecting.} If
the attacker intercepts the traffic \emph{towards} a node, she can instead
corrupt the content of the \textsf{BLOCK} message such that the victim considers
it invalid. In both cases, the recipient of the blocks remains uninformed for 20
minutes. \remove{Modifying these messages is easy and could even be done in the
data plane (Section~\ref{sec:delay_results}).}

\remove{
Unlike partitioning attacks, an attacker can delay the overall propagation of
the blocks towards a node even if she intercepts a subset of the node
connections. Intercepting even one direction one of one connection is enough as
long as this connection sees some \textsf{BLOCK} or \textsf{GETDATA} of a block.
}

\myitem{Example} As an illustration, consider Fig.~\ref{fig:delay_overview}, and assume that AS8 is the attacker and $C$, the
victim. Suppose that $A$ and $B$ both advertise a block (say, block $X$)
to $C$ via an \textsf{INV} message and that, without loss of generality, the
message from $A$ arrives at $C$ first. $C$ will then send a \textsf{GETDATA}
message back to $A$ requesting block $X$ and start a 20 minute timeout count. By modifying the content of the \textsf{GETDATA} node $A$ receives, AS8 indirectly controls what node $A$ will send to node $C$.
This way the attacker can delay the delivery of the block by up to 20 minutes while avoiding detection and disconnection.
Alternatively, AS8 could modify the \textsf{BLOCK} message.

\myitem{Practicality} We verified the practicality of delay attacks by
implementing an interception software which we used
against our own Bitcoin nodes.
We show that intercepting 50\% of a node connections is enough to keep the node uninformed for 63\% of its uptime (Section~\ref{sec:delay_results}). 

We also evaluated the impact that ASes, which are naturally
traversed by a lot of Bitcoin traffic, could have on the network using a
scalable event-driven simulator. We found that due to the relatively high
degree of multi-homing that pools employ, only very powerful coalitions of network
attackers (e.g., all ASes based in the US) could perform a network-wide delay attack. Such an attack is thus unlikely to occur in practice.
\remove{Particularly, all US-based ASes can increase the fork rate to 30\% from the
normal 1\%.} 

\myitem{Impact} Similarly to partitioning attacks, the impact of a delay attack
depends on the number and type (e.g., pool gateway)
of impacted nodes. At the node-level, delay attacks can keep the
victim eclipsed, essentially performing a denial of service attack or rendering
it vulnerable to 0-confirmation double spends. If the node is a gateway of a
pool, such attacks can be used to engineer block races, and waste the mining power of the pool. Network-wide attacks
increase the fork rate and render the network vulnerable to other exploits. If a sufficient number of blocks are
discarded, miners revenue is decreased and the network is more vulnerable
to double spending. A slowdown of block transmission
can be used to launch selfish mining attacks by adversaries with mining power.

\section{Partitioning Bitcoin}
\label{sec:partition}


\begin{figure*}[t]
	\centering
	\begin{subfigure}[t]{.22\textwidth}
  		\includegraphics[width=\textwidth]{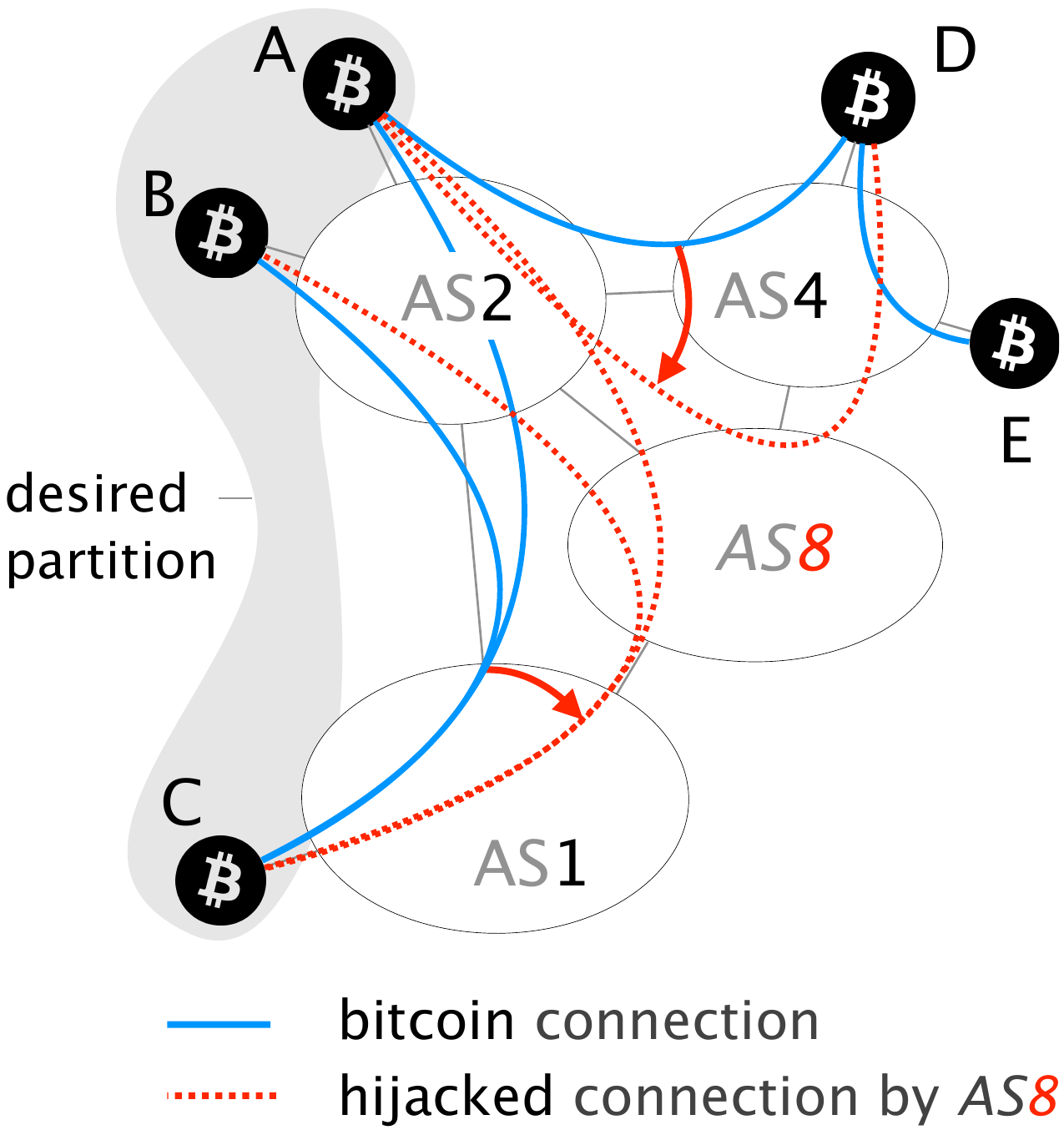}
		\caption{Feasible partition}
		\label{fig:visible}
	\end{subfigure}
	\qquad \qquad 
	\begin{subfigure}[t]{.22\textwidth}
		\includegraphics[width=\textwidth]{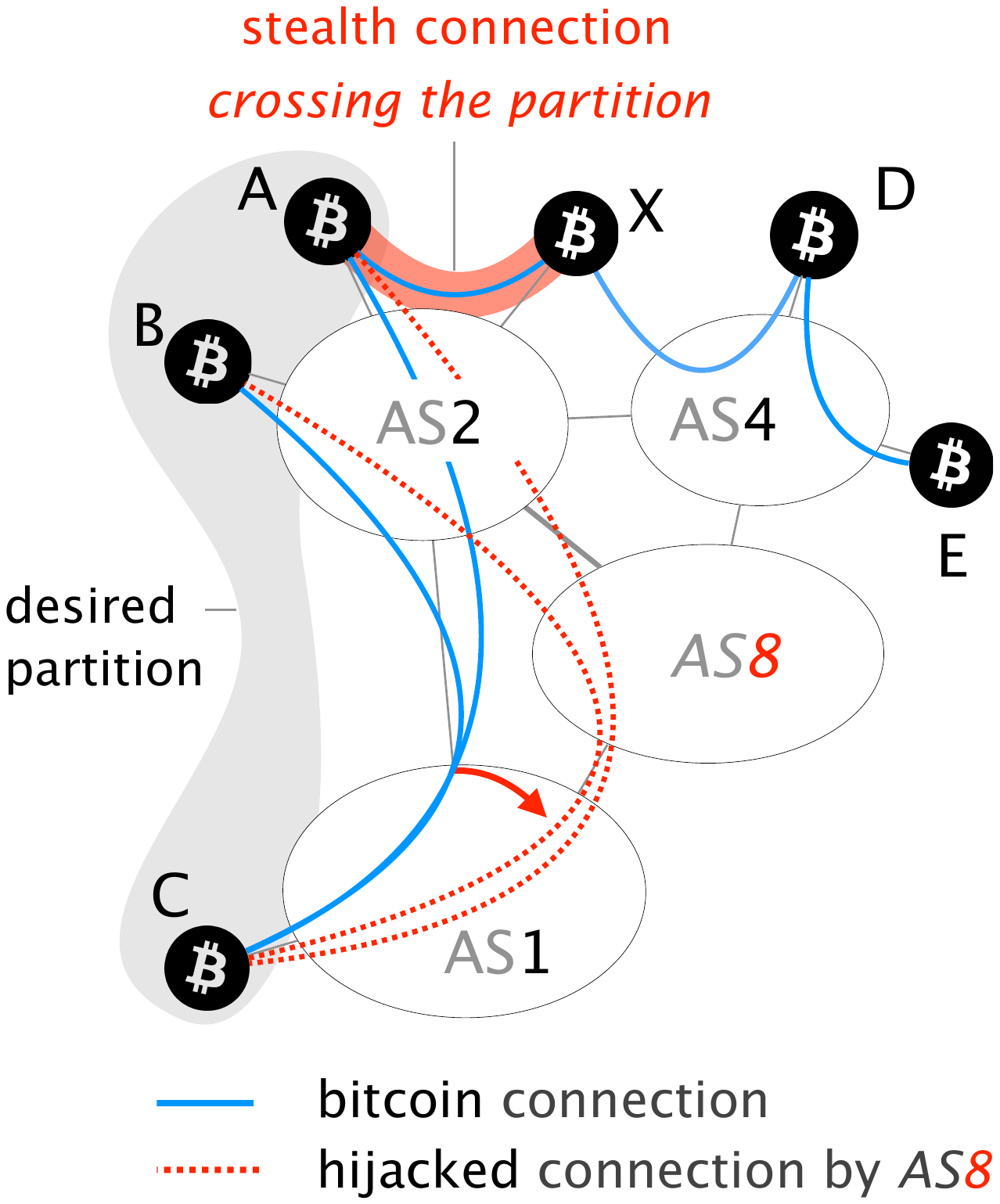}
		\caption{Infeasible partition because of intra-AS connections }
		\label{fig:intraAS}
	\end{subfigure}
	\qquad \qquad
	\begin{subfigure}[t]{.30\textwidth}
		\includegraphics[width=\textwidth]{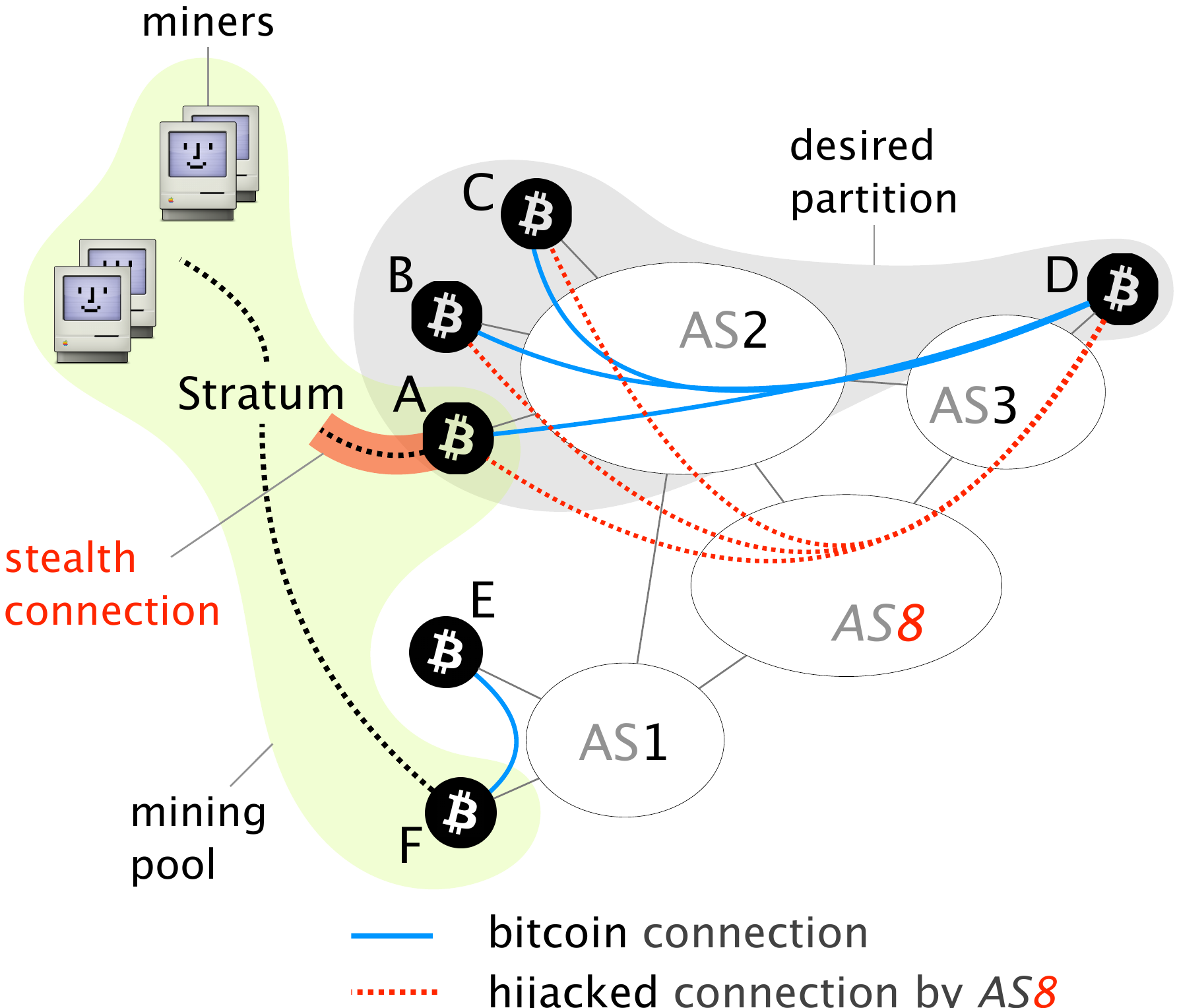}
		\caption{Infeasible partition because of intra-pool connections}
		\label{fig:intrapool}
	\end{subfigure}
  \caption{Not all Bitcoin connections can be diverted by an attacker implying that some partitions cannot be formed.}
 \label{fig:partition_possible}
\end{figure*}

In this section, we elaborate on partition attacks in which an AS-level adversary seeks to isolate a set of nodes $P$. We first describe which partitions are feasible by defining which connections may cause information leakage (Section~\ref{ssec:hidden}) to the isolated set.  We then discuss how an attacker may better select a $P$ that is feasible,
if she has some view of the Bitcoin topology (Section~\ref{sssec:plan}).
Next, we walk through the entire attack process, starting with the interception of Bitcoin traffic, the detection of leakage points and the adaptation of $P$ until the partition is successfully created (Section~\ref{ssec:disconnection}). In particular, we present an algorithm which, given a set of nodes $P$, leads the attacker to isolate the maximal feasible subset. Finally, we prove that our algorithm is correct (Section~\ref{ssec:proof}).

\remove{As well as why the latter are 
The main factor that makes the attack challenging is the existence of connections that the AS-level adversary does not see and that he cannot cut. If such connections cross the partition, information cannot be prevented from flowing from one component to the other in which case the partition is infeasible. Thus, the attacker should detect these connections and modify $P$ into a $P'$ (subset of P) such that they are no such connections bridging the partition. 
}

\remove{
\begin{algorithm}[t!]
	\small
	\SetKwInOut{Input}{Input}
	\KwIn{\textcolor{mygray}{-} $(G_1,G_2)$, a partition of the known Bitcoin IP addresses; and \newline 
	\textcolor{mygray}{-} $P = (P_1, \cdots, P_n)$, the sets of known Bitcoin IP addresses acting as gateways for each mining pool.}
	\KwOut{\emph{true} or \emph{false}}
	\Begin{
			\If{$AS(G_1) \cap AS(G_2) \neq \varnothing$}{
				\Return{false}\,\;
			}
			\For{$P_i \in P$}{
				\If{$\neg (AS(P_i) \cap AS(G_1) \oplus AS(P_i) \cap AS(G_2))$}{
					\Return{\text{false}}\,\;
				}
			}
			\Return{true}\,\;
	}
\caption{Algorithm testing whether a partition is feasible given a view of the Bitcoin topology.}
\label{alg:is_partition_possible}
\end{algorithm}
}

\subsection{Characterizing feasible partitions}
\label{ssec:hidden}

An attacker can isolate a set of nodes $P$ from the network if and
only if \emph{all} connections $(a,b)$ where $a \in P$ and $b \not\in P$
can be intercepted. We refer to such connections as \emph{vulnerable} and to
connections that the attacker cannot intercept as \emph{stealth}.

\myitem{Vulnerable connections:} A connection is \emph{vulnerable} if: \emph{(i)} an
attacker can divert it via a BGP hijack; and \emph{(ii)} it uses the Bitcoin protocol. The first requirement enables the
attacker to intercept the corresponding packets, while the second enables her to
identify and then drop or monitor these packets. 

As an illustration, consider Fig.~\ref{fig:visible}, and assume that the
attacker, AS8, wants to isolate $P=\{A,B,C\}$. By hijacking the prefixes
pertaining to these nodes the attacker receives all traffic from nodes $A$ and
$B$ to node $C$, as well as the traffic from node $D$ to node $A$. The path the
hijacked traffic follows is depicted with red dashed lines and the original path
with blue lines. As all nodes communicate using the Bitcoin protocol, their
connections can easily be distinguished, as we explain in
Section~\ref{ssec:disconnection}. Here, AS8 can partition $P$ from the rest of the network by dropping the connection from node $D$ to node $A$.

\myitem{Stealth connections:} A connection is \emph{stealth} if the attacker cannot intercept it. 
We distinguish three types of stealth connections:
\emph{(i)} intra-AS; \emph{(ii)} intra-pool; and \emph{(iii)} pool-to-pool. 

\noindent\textit{\textbf{intra-AS:}} An attacker cannot intercept
connections within the same AS using BGP hijack. Indeed, internal traffic does
not get routed by BGP, but by internal routing protocols (e.g., OSPF, EIGRP).
Thus, any intra-AS connection crossing the partition border renders the partition infeasible. Such connections represent only 1.14\% of all the possible connection nodes can create (this percentage is calculated
based on the topology we inferred in Section~\ref{sec:topology}).

As an illustration, consider Fig.~\ref{fig:intraAS} and assume that
the attacker, AS8, wants to isolate $P=\{A,B,C\}$. By hijacking the
corresponding BGP prefixes, AS8 can intercept the connections
running between nodes $A$ and $B$ to node $C$. However, she does not intercept the intra-AS
connection between $A$ and $X$. This means that node $X$ will inform node $A$
of the blocks mined in the rest of the network, and node A
will then relay this information further within $P$. Thus, $P=\{A,B,C\}$ is not
feasible. Yet, observe that isolating $I=\{B,C\}$ is possible. In the
following, we explain how the attacker can detect that $A$ maintains a stealth
connection leading outside of the partition and dynamically adapt to
isolate $I$ instead.


\noindent\textit{\textbf{intra-pool:}} Similarly to intra-AS connections, an attacker
 might not be able to cut connections between gateways belonging to the same mining pool.
 This is because mining pools might rely on proprietary or even encrypted protocols for internal communication. 

As an illustration, consider Fig.~\ref{fig:intrapool} and assume that the attacker, AS8, wants to isolate $P=\{A,B,C,D\}$. By hijacking the corresponding prefixes, she would intercept and cut all Bitcoin connections between nodes A, B, C, D and nodes E, F. However, nodes A and F would still be connected internally as they belong to the same (green) pool.
Again, observe that while isolating $P=\{A,B,C, D\}$ is not feasible, isolating $I=\{B,C,D\}$ from the rest of the network is possible.

\remove{
\textit{encrypted intra-pool or inter-pool connections:} An attacker
might not be able to cut connections between nodes that are using an encrypted private channel between them. These may belong to the same mining pool, or to different pools that have agreed to share such a private connection.
This is because 
such connections may simply be difficult to identify. 
 As such,
if a pool uses encrypted tunnels to connect between its multiple gateways or to another pool that is outside of $P$, this connection would act as a bridge, rendering the partition
infeasible.

}

\noindent\textit{\textbf{pool-to-pool:}} Finally, an attacker cannot
intercept (possibly encrypted) private connections, corresponding to peering agreements between mining pools. From
the attacker's point of view, these connections can be treated as intra-pool
connections and the corresponding pair of pools can be considered as one larger pool. Note that such connections are different than public initiatives to interconnect all pools, such as the Bitcoin relays~\cite{relays}. Unlike private peering agreements, relays \emph{cannot} act as bridges to the partition (see Appendix~\ref{ssec:faq}).

\remove{
Consider Fig.~\ref{fig:partition_possible}, and assume that the attacker, AS8, wants to
isolate $P=\{A,B,C\}$ in each case. The path the hijacked traffic follows is
depicted with red dashed lines, the original path with blue lines.
In the first case (Fig.~\ref{fig:visible}), all nodes
communicate using the Bitcoin protocol, and connections can easily be intercepted and
distinguished as we explain in \ref{ssec:disconnection}. Here, AS8 can partition $P$ from the rest of the network by simply
dropping the connection from node A to node D.
In the second case (Fig.~\ref{fig:intraAS}) the attacker fails to intercept the intra-AS
connection between $A$ and $X$, and thus fails to establish the partition (observe that here isolating $I=\{B,C\}$ is still possible). Finally, in Fig.~\ref{fig:intrapool} nodes A and F maintain an intra-pool encrypted connection that the attacker fails to identify (the attacker may still isolate $P'=\{B,C,D\}$ in this case). 
}

\subsection{Preparing for the attack}
\label{sssec:plan}
In light of these limitations the attacker can apply two techniques to avoid having stealth connections crossing the partition she creates. First, she can
include in $P$ either all or none of the nodes of an AS, to avoid intra-AS connections
crossing the partition. This can be easily done as the mapping from IPs to ASNs is
publicly available~\cite{caida_ip_to_asn}. Second, she can include in $P$
either all or none of the gateways of a pool, to avoid intra-pool connections crossing
the partition. Doing so requires the attacker to know all the gateways of the
mining pools she wants to include in $P$. Inferring the gateways is outside the
scope of this paper, yet the attacker could use techniques described in~\cite{miller2015discovering} and leverage her ability to inspect the traffic
of almost every node via hijacking (see Appendix \ref{ssec:gw}). Even with the above measures, $P$ may still contain leakage points that the attacker will need to identify and exclude (see below). Yet, these considerations increase the chances of establishing the desired partition as well as reducing the time required to achieve it.

\remove{In any case, due to the hardness of acquiring a perfect topological view and to
the dynamic nature of the Bitcoin network the attacker would need to adapt the partition
during the attack as described in the next Subsection. Even so, a good approximation
of a feasible partition would lead to a more effective attack as the $P'$ will
}

\begin{algorithm}[t]
	\small
	\SetKwInOut{Input}{Input}
	\KwIn{\textcolor{mygray}{-} $P$, a set of Bitcoin IP addresses to disconnect from the rest of the Bitcoin network; and \newline 
	\textcolor{mygray}{-} $S = [pkt_1, \cdots]$, an infinite packet stream of \newline diverted Bitcoin traffic resulting from the hijack of the prefixes pertaining to $P$.}
	\KwOut{False if there is no node $\in P$ that can be verifiably isolated;}
	\textbf{enforce\_partition($P,S$)}:
	
	\Begin{
			$U \gets \varnothing$\,\;
			$L \gets \varnothing$\,\;
			\While{$P \setminus (L \cup U) \neq \varnothing$}{
				\For{$pkt \in S$}{
					\If{$pkt.ip\_src \in P \wedge pkt.ip\_src \notin L $}{ \label{line:is_node_inside}
						$last\_seen[pkt.ip\_dst] = now()$\,\;
						$U \gets U \setminus \{ pkt.ip\_src \}$ \,\;\label{line:add_inactive}
						$detect\_leakage(U, pkt)$\,\;\label{line:detect_leakage}
						
					}
					\Else{
						$drop(pkt)$\,\; \label{line:drop-visible}
					}
					
					\For{$src \in P \wedge src \notin L$}{
						\If{$last\_seen[src] > now() - threshold$}{
							$U \gets U \cup \{ src \}$ \label{line:removal_inactive}
						}
					}
				}
			}
			\Return{false}\,\;
	}
\caption{Partitioning algorithm.}
\label{alg:attack_algorithm}
\end{algorithm}

\begin{algorithm}[t]
	\small
	\SetKwInOut{Input}{Input}
	\SetKwFunction{FRecurs}{FnRecursive}
	\KwIn{\textcolor{mygray}{-} $U$, a set of Bitcoin IP addresses the attacker cannot monitor; and \newline 
	\textcolor{mygray}{-} $pkt$, a (parsed) diverted Bitcoin packet.}
	\textbf{detect\_leakage($U, pkt$)}:
	
	\Begin{
		\If{$contains\_block(pkt) \vee contains\_inv(pkt)$}{
			\If{$hash(pkt) \in Blocks(\neg (P \setminus L))$}{
				$L \gets L \cup \{pkt.ip\_src \}$\,\; \label{line:leakage_detection}
				$drop(pkt)$\,\; \label{line:drop_block}
			}
		}

	}
\caption{Leakage detection algorithm.}
\label{alg:leakage_algorithm}
\end{algorithm}

\subsection{Performing the attack}
\label{ssec:disconnection}

We now describe how a network adversary can successfully perform a partitioning
attack. The attack is composed of two main phases: \emph{(i)} diverting relevant Bitcoin traffic; and \emph{(ii)} enforcing the partition. In the
former phase, the adversary diverts relevant Bitcoin traffic using BGP
hijacking. In the latter phase, the attacker cuts all vulnerable connections
that cross the partition and excludes from $P$ nodes which are identified as leakage points.
Leakage points are nodes that are connected to the rest of the network via stealth connections.

\myitem{Intercept Bitcoin traffic:} The attacker starts by hijacking all the
prefixes pertaining to the Bitcoin nodes she wants to isolate, i.e. all the
prefixes covering the IP addresses of nodes in $P$. As a result, she receives
all the traffic destined to these prefixes, which she splits into two packet
streams: relevant and irrelevant. Relevant traffic includes any Bitcoin
traffic destined to nodes in $P$. This traffic should be
further investigated. Irrelevant traffic corresponds to the remaining traffic
which should be forwarded back to its legitimate destination. 

To distinguish between relevant and irrelevant traffic, the attacker applies
a simple filter matching on the IP addresses, the transport protocol and ports
used, as well as certain bits of the TCP payload. Specifically, the attacker
first classifies as irrelevant all non-TCP traffic as well as all traffic with
destination IPs which are not included in $P$. In contrast, the attacker
classifies as relevant all traffic with destination or source TCP port the
default Bitcoin port (8333)\remove{and other publicly advertised port numbers (e.g.
via \textsf{ADDR} messages)}. Finally, she classifies as relevant all packets
which have a Bitcoin header in the TCP payload. Any remaining traffic is considered irrelevant.

\myitem{Partitioning algorithm:}
Next, the attacker processes the relevant traffic according to Algorithms~\ref{alg:attack_algorithm} and~\ref{alg:leakage_algorithm}. We start by presenting their goal before describing them in more details.

The high-level goal of the algorithms is to isolate as many nodes in $P$ as possible. To do so, the algorithms identify $L$, the nodes that are
leakage points, and disconnect them from the other nodes in $P$. Also,
the algorithms maintain a set of verifiably isolated nodes  $P'= P\setminus \{U\cup L\}$, 
where $U$ corresponds to the nodes that cannot be monitored (e.g., because they never
send packets). In particular,  Algorithm~\ref{alg:leakage_algorithm} is in charge of identifying $L$,
while Algorithm~\ref{alg:attack_algorithm} is in charge of identifying $U$ and
performing the isolation itself. 

\remove{
The algorithms
also identify $U$, the nodes that cannot be monitored (e.g., because they never
send packets). It then provably isolate the nodes in $P \setminus L \setminus
U$. Algorithm~\ref{alg:leakage_algorithm} is in charge of identifying $L$,
while Algorithm~\ref{alg:attack_algorithm} is in charge of identifying $U$ and
performing the isolation itself.
}

We now describe how the algorithms work. Algorithm~\ref{alg:attack_algorithm}
starts by initializing $L$ and $U$ to $\varnothing$. For every received packet,
the algorithm first decides whether the packet belongs to a connection internal
to $P \setminus L$ or to one between a node in $P \setminus L$ and an external
node based on the source IP address. If the source IP is in $P \setminus L$, the packet belongs to an
internal connection and it is given to Algorithm~\ref{alg:leakage_algorithm}
to investigate if the corresponding node acts as a leakage point  (Algorithm~\ref{alg:attack_algorithm},
Line~\ref{line:detect_leakage}). Otherwise, the packet belongs to a connection that crosses the
partition and is dropped (Algorithm~\ref{alg:attack_algorithm},
Line~\ref{line:drop-visible}).

Given a packet originated from $P \setminus L$,
Algorithm~\ref{alg:leakage_algorithm} checks whether the sender of the packet
is advertising information from outside of $P \setminus L$. Particularly, the
attacker checks whether the packet contains an \textsf{INV} message with the
hash of a block mined outside of $P \setminus L$ (or the block itself). If it
does so, the sender must have a path of stealth connections to a node outside of $P
\setminus L$ from which the block was transmitted. Thus the sender is a leakage
point and is added to $L$ (Algorithm~\ref{alg:leakage_algorithm},
Line~\ref{line:leakage_detection}). The actual packet is also dropped to
prevent this information from spreading.

To detect whether a node in $P \setminus L$ is a leakage point, an attacker
should be able to intercept at least one of that node's connections.
Specifically, the node should have a vulnerable connection to another node
within $P \setminus L$, so that the attacker can monitor the blocks it
advertises. To keep track of the nodes that the attacker cannot monitor,
Algorithm~\ref{alg:attack_algorithm} maintains a set $U$ which contains the
nodes she has not received any packets from for a predefined time threshold.
(Algorithm~\ref{alg:attack_algorithm}, Line~\ref{line:removal_inactive}).
Whenever one of these nodes manages to a establish a connection that the
attacker intercepts, it is removed from $U$ (Algorithm~\ref{alg:attack_algorithm}, Line~\ref{line:add_inactive}).

\myitem{Example:} We now show how the algorithms work on the example of Fig.~\ref{fig:intraAS} in which the attacker, AS8, aims to isolate $P=\{A,B,C\}$. By hijacking the prefixes corresponding to these nodes, the attacker intercepts the connections $(B,C)$ and $(A,C)$ and feeds the relevant packets to the algorithms. Recall that the partition is bridged by a stealth (intra-AS) connection between nodes $A$ and $X$ which cannot be intercepted by the attacker. When a block outside $P$ is mined, node $X$ will inform $A$ which then will advertise the block to $C$. The attacker will catch this advertisement and will conclude that node $A$ is a leakage point. After that, the attacker will drop the packet and will add $A$ to $L$. As such, all future packets from $A$ to other nodes within $P \setminus L =\{A,B\} $ will be dropped. Observe that the partition isolating $P \setminus L=\{B,C\}$ is indeed the maximum  feasible subset of $P$.

\remove{In Fig.~\ref{fig:intrapool}, the attacker, AS8, aims at isolating
$P=\{A,B,C,D\}$. The partition is bridged by a stealth (intra-pool) connection between nodes $A$ and $E$ which belong to the green pool. If the green pool learns about a block mined outside $P\setminus L$ by one of its gateways, namely node $F$, the pool will inform its other gateway, node $A$ about it. Node $A$ will then advertise this block to $D$. The attacker will catch this advertisement and will conclude that node $A$ is a leakage point. In this case, the attacker will again remove node $A$ from $P'$ and will be left with $P'=\{B,C,D\}$, which is again feasible.
}

\subsection{Correctness of the partitioning algorithm}
\label{ssec:proof}

We now prove the properties of Algorithm~\ref{alg:attack_algorithm}.

\begin{theorem}
Given $P$, a set of nodes to disconnect from the Bitcoin network, there exists a unique maximal subset $I\subseteq P$ that can be isolated. Given the assumption that Bitcoin nodes advertise blocks that they receive to all their peers,
Algorithm~\ref{alg:attack_algorithm} isolates all nodes in $I$, and maintains a set $P' = P\setminus \{U\cup L\}\subseteq I$ that contains all nodes in $I$ that have a monitored connection and are thus known to be isolated. 
\end{theorem}

\begin{proof}
	Consider the set of nodes $\mathcal{S}\subseteq P$ that has a path of stealth connections to some nodes not in $P$. 
	Clearly, nodes in $\mathcal{S}$ cannot be isolated from the rest of the network by the attacker. 
	Let $I = P\setminus \mathcal{S}$. Notice that $I$ is the maximal set in $P$ that can be disconnected by an attacker. 
	Now, notice that every node in $S$ is placed in sets $L$ or $U$ by the algorithm: 
	if the node has a monitored connection and is caught advertising external blocks it is placed in $L$ (Algorithm~\ref{alg:leakage_algorithm} Line~\ref{line:leakage_detection}).
	If it is not monitored then it is placed in $U$ (Algorithm~\ref{alg:attack_algorithm}, Line~\ref{line:removal_inactive}).
	
	Notice also that the entire set $I$ is isolated from the network. If some node has no stealth connection outside, and was removed solely for the lack of monitoring, it is still having all its packets from outside of $P \setminus L$ dropped -- Algorithm~\ref{alg:attack_algorithm} Line~\ref{line:drop-visible}).  

\end{proof}

\section{Delaying Block Propagation}
\label{sec:delay}

While partitioning attacks (Section~\ref{sec:partition}) are particularly effective and can be performed by
 any AS, they require full control over the victim's traffic and are also highly visible. In this section, we
explore \emph{delay attacks}, which can cause relatively severe delays in block propagation, even when an attacker intercepts only one of the victim's connections, and wishes the attack to remain relatively undetectable. 

In this attack, the adversary delays the delivery of a block by modifying the content of specific messages.
This is possible due to the lack of encryption and of secure integrity checks of Bitcoin messages. In addition to these, the attacker leverages the fact that nodes send block requests to the first peer that advertised each block and wait 20 minutes for its delivery, before requesting it from another peer. 

The first known attack leveraging this 20 minutes timeout~\cite{Gervais:2015:TDB:2810103.2813655} mandates the adversary to be connected
to the victim and to be the first to advertise a new block. After a successful
block delay, the connection is lost. In contrast, network-based delay attacks are
more effective for at least three reasons: \emph{(i)} an attacker can act on
existing connections, namely she does not need to connect to the victim which
is very often not possible (e.g, nodes behind a NAT); \emph{(ii)} an attacker
does not have to be informed about recently mined blocks by the
victim's peers to eclipse it; and \emph{(iii)} the connection that was used for
the attack is not necessarily lost, prolonging the attack.

Particularly, the effectiveness of the delay attack depends on the direction and
fraction of the victim's traffic the attacker intercepts. 
Intuitively, as Bitcoin clients request blocks from one peer at a time, the probability that the attacker will intercept such a connection increases proportionally with the fraction of the connections she intercepts. In addition, Bitcoin connections are bi-directional TCP connections, meaning the attacker may intercept one 
direction (e.g., if the victim is multi-homed), both, or none at all. Depending on the direction she intercepts, the attacker fiddles with different 
messages. In the following, we explain the mechanism that is used to perform the attack if the attacker intercepts traffic \emph{from} the victim  (Section~\ref{ssec:outgoing_attack}) or \emph{to} the victim node (Section~\ref{ssec:incoming_attack}). While in both cases the attacker does delay block propagation for 20 minutes, the former attack is more effective. \remove{The node-level attack can be easily extended to a network level attack, in which case the adversary delays all blocks propagated through the connections she intercepts.}

\begin{figure}[t]
 \centering
 \begin{subfigure}[t]{0.45\columnwidth}
 \includegraphics[width=\columnwidth]{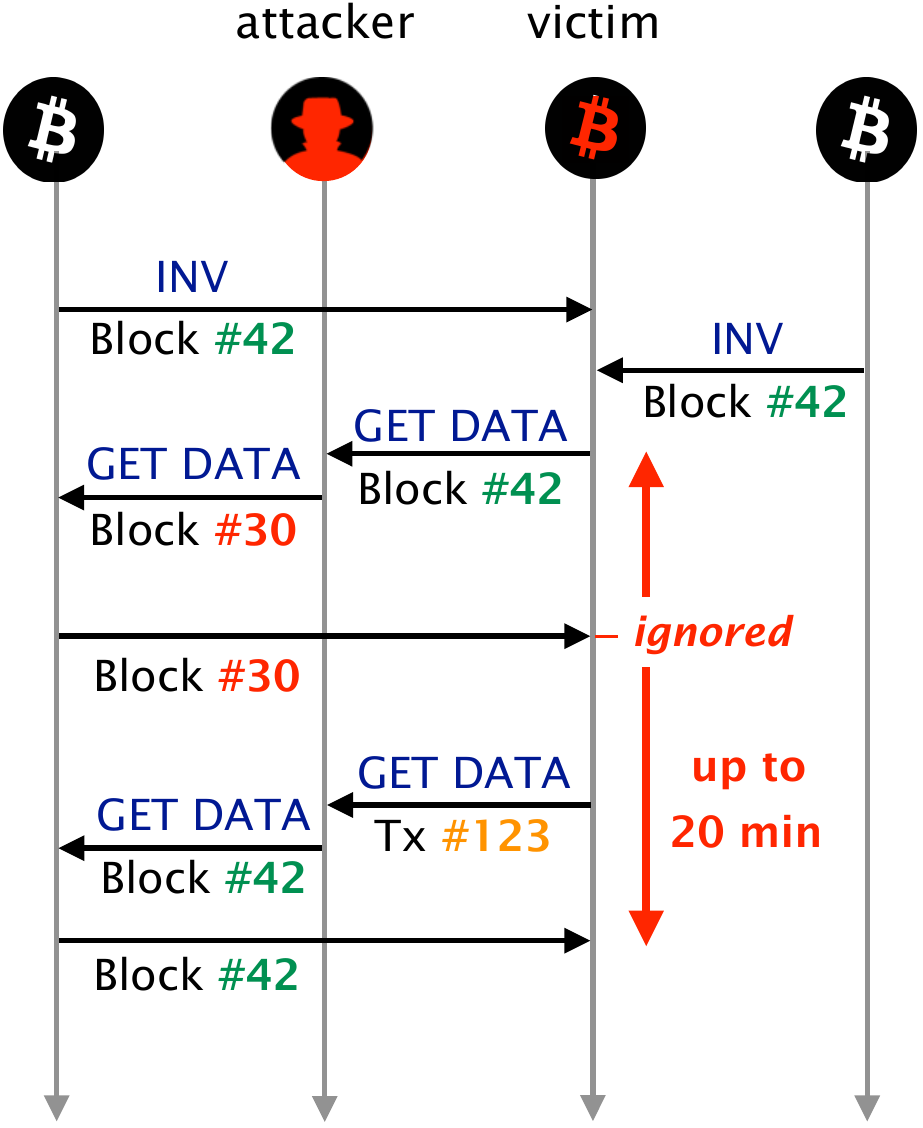}
 \caption{$\curvearrowleft$ Attacker $\curvearrowleft$ victim}
 \label{fig:delay_attack_left}
 \end{subfigure}
 \hfill
 \vrule
 \hfill
 \begin{subfigure}[t]{0.45\columnwidth}
 \includegraphics[width=\columnwidth]{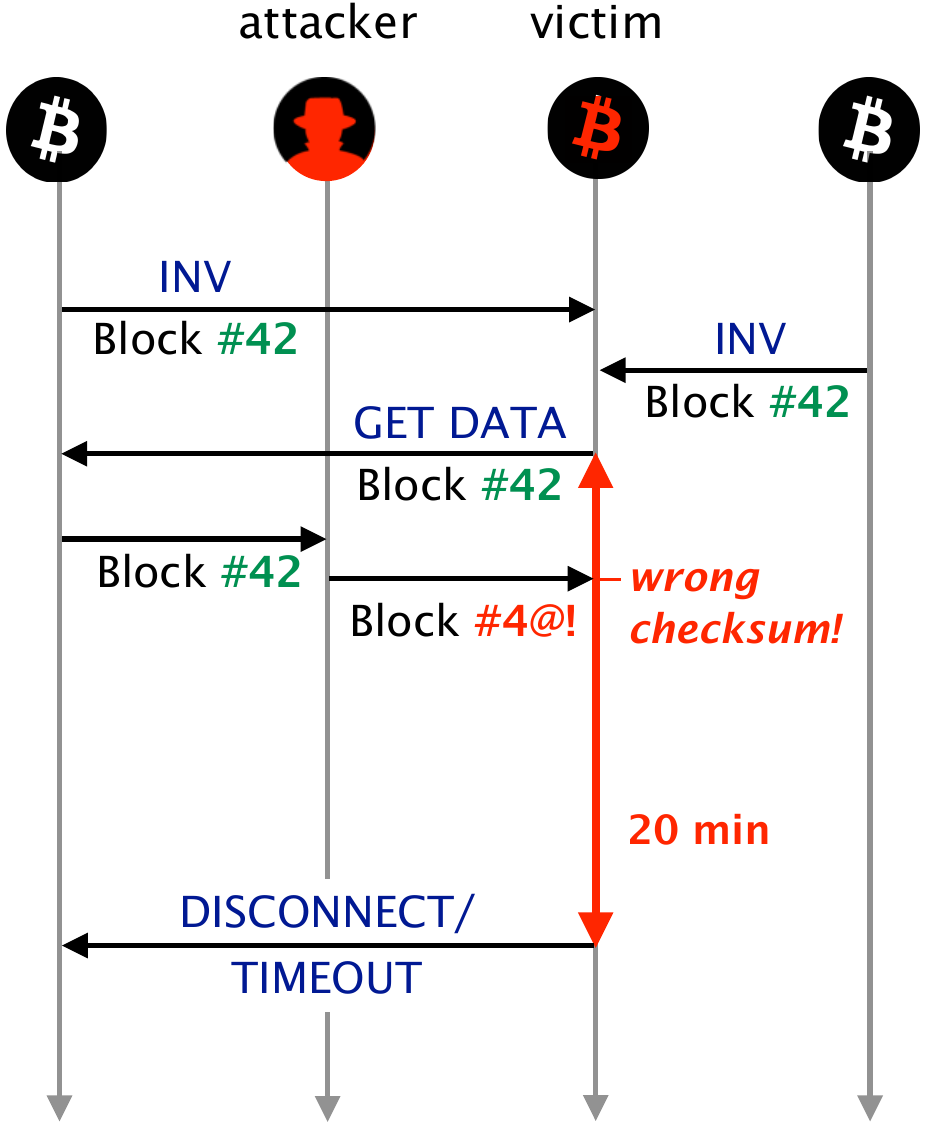}
 \caption{$\curvearrowright$ Attacker $\curvearrowright$ victim}
 \label{fig:delay_attack_right}
 \end{subfigure}
 \caption{An attacker can perform different delay attacks depending on the direction of traffic she intercepts. When on the path from the victim, it can modify the \textsf{GETDATA} message (a), while it can modify the \textsf{BLOCK} message when intercepting the opposite path (b).}
 \label{fig:delay_attack_figure}
\end{figure}

 \remove{

We now explain the mechanism of attacking a single node. This mechanism can be easily extended to a network level attack, in which case the adversary, delays all blocks propagated through the connections she intercepts.
Particularly, the effectiveness of the attack depends on the direction and
fraction of the victim's traffic the attacker intercepts. Bitcoin nodes
establish connections with many peers. To effectively delay the delivery of a block to a victim node the
attacker should intercept the connection between the victim and the first of its peers which advertised this particular block to it. 
This is due to the fact that a bitcoin client only requests a block from the first peer which advertised it.
Intuitively, the largest the fraction of connections the attakcer intercepts the higher her chance to delay the delivery of a blcok to it
All these connection are bi-directional TCP connections, meaning the attacker can intercept one direction (e.g., if the victim is multi-homed), both or none at all. We start by describing the most effective attack, which can be performed by an
adversary that intercepts at least a fraction of the traffic that originates
\emph{from} the victim (\S\ref{ssec:outgoing_attack}).
We then describe another attack that can be performed by a MITM observing at least a 
fraction of the traffic destined \emph{to} the victim (\S\ref{ssec:incoming_attack}). Both
attacks allow  an adversary to delay blocks by 20 minutes, but the
former is less visible.
}

\subsection{The attacker intercepts outgoing traffic}
\label{ssec:outgoing_attack}

Once a node receives a notification that a new
block is available via an \textsf{INV} message, it issues a download request
to its sender using a \textsf{GETDATA} message. As an illustration in Fig.~\ref{fig:delay_attack_left} the victim requests Block 42
from the first of its peers that advertised the block. 
Since the attacker intercepts the traffic from the 
victim to this peer, she can modify this \textsf{GETDATA} message to request an older block, instead of the original one.
In Fig.~\ref{fig:delay_attack_left} for example, the attacker replaces the hash corresponding to block 42 with that of block 30. The advantage of doing so, over just dropping the packet, is that the length of the message is left unchanged. Notice that if the packet was dropped the attacker would need to intercept both directions of the connection and update the TCP sequence numbers of all following packets to ensure that the connection is not dropped. If the block is not delivered, the victim node will disconnect after 20 minutes.
 To avoid a disconnection, the attacker can use another
\textsf{GETDATA}, sent within the 20 minute window, to perform the reverse
operation. Specifically, she modifies the hash back to the original one, requested by the victim. 
Since the \textsf{GETDATA} message for blocks and transactions have the same structure, the attacker is more likely
to use the latter as these are much more common. In Fig.~\ref{fig:delay_attack_left} for example, she changes the hash of
the transaction (Tx \#123) to the hash of block 42. Since the block is delivered within the timeout neither of the nodes disconnects or has any indication of the attack (e.g., an error in the log files).

\subsection{The attacker intercepts incoming traffic}
\label{ssec:incoming_attack}

We now describe the mechanism an attacker would use if she intercepts traffic
towards the victim, \emph{i.e.} she can see messages received by the victim, but not the messages that it sends.
This attack is less effective compared to the attack working in the opposite direction, 
as it will 
eventually result in the connection being dropped 20 minutes after the first delayed block (similarly to~\cite{Gervais:2015:TDB:2810103.2813655}).
\remove{Yet, this attack still allows a MITM intercepting only this direction to effectively delay the delivery of a block to a node for 20 minutes.}
In this case, the attack focuses on the
\textsf{BLOCK} messages rather than on the \textsf{GETDATA}. A naive attack
would be for the attacker to simply drop any \textsf{BLOCK} message she sees. As Bitcoin
relies on TCP though, doing so would quickly kill the TCP connection. A
better, yet still simple approach is for the attacker to corrupt the contents of a
\textsf{BLOCK} message while preserving the length of the packet (see Fig.~\ref{fig:delay_attack_right}). This simple operation causes the \textsf{BLOCK} to be
discarded when it reaches the victim, because of a checksum mismatch. Surprisingly
though, we discovered (and verified) that the victim will \emph{not} request the block again, be
it from the same or any other peer. After the 20 minute timeout elapses,
the victim simply disconnects because its requested block did not arrive on time.

An alternative for the adversary is to replace the hash of the most recent
\textsf{Block} with a hash of an older one in all the \textsf{INV} messages the
victim receives. This attack however would fail if the attacker intercepts only
a fraction of the connections, as the victim will be informed via other connections. As such, this practice is only useful when the attacker
hijacks and thus intercepts all the traffic directed to the victim.

\section{How vulnerable is Bitcoin to routing attacks? A comprehensive measurement analysis}
\label{sec:topology}

Evaluating the impact of routing attacks requires a good understanding of the
routing characteristics of the Bitcoin network. In this section, we explain the
datasets and the techniques used to infer a combined Internet and Bitcoin
topology (Section~\ref{ssec:topo_datasets}). We then discuss our key findings and their impact on the effectiveness of the two routing attacks we consider (Section~\ref{ssec:topo_findings}). 

\subsection{Methodology and datasets}
\label{ssec:topo_datasets}

Our study is based on three key datasets: \emph{(i)} the IP addresses
used by Bitcoin nodes and gateways of pools; \emph{(ii)} the portion of mining power each pool possesses; \emph{(iii)} the forwarding path  taken between any two IPs.
While we collected these datasets over a period of 6 months, starting from
October 2015 through March 2016, we focus on the results from a 10
day period starting from November 5th 2015, as the results of our analysis do not
change much through time.

\remove{
\begin{figure}[t]
 \centering
 \includegraphics[width=.65\columnwidth]{figures/method_topo}
 \caption{Summary of our methodology for inferring the network characteristics of the Bitcoin topology based on network and Bitcoin data.}
 \label{fig:method}
\end{figure}
}


\myitem{Bitcoin IPs} We started by collecting the IPs of regular nodes (which
host no mining power) along with the IPs of the gateways the pools use to connect to the network. We gathered this dataset
by combining information collected by two Bitcoin supernodes with publicly available
data regarding mining pools. One supernode was connected to {\raise.17ex\hbox{$\scriptstyle\sim$}}2000 Bitcoin nodes per day, collecting  block propagation information, while the other was crawling the Bitcoin network, collecting the {\raise.17ex\hbox{$\scriptstyle\sim$}}6{,}000 IPs of active nodes each day.

We inferred which of these IPs act as the gateway of a pool in two steps.
\emph{First}, we used block propagation information (gathered by the first supernode), considering that the gateways of a pool are most likely the
first to propagate the blocks this pool mines. Particularly, we
assigned IPs to pools based on the timing of the \textsf{INV} messages received
by the supernode. We considered a given IP to belong to a gateway
of a pool if: \emph{(i)} it relayed blocks of that pool more than once during
the 10 day period; and \emph{(ii)} it frequently was the first to relay a block of
that pool (at least half as many times as the most frequent node for
that pool). \emph{Second}, we also considered as extra gateways the IPs of the
stratum servers used by each mining pool. Indeed, previous
studies~\cite{miller2015discovering} noted that stratum servers tend to be
co-located in the same prefix as the pool's gateway. Since the URLs of the
stratum servers are public (Section~\ref{sec:background}), we simply resolved
the DNS name (found on the pools websites or by directly connecting to them)
and add the corresponding IPs to our IP-to-pool dataset.

\myitem{Mining power} To infer the mining power attached to pools, we
tracked how many blocks each pool mined during the 10 days
interval~\cite{blockchain} and simply assigned them a proportional share of the
total mining power.

\myitem{AS-level topology and forwarding paths} We used the AS-level topologies
provided by CAIDA~\cite{caida_as_level} to infer the forwarding paths taken
between any two ASes. An AS-level topology is a directed graph in which a
node corresponds to an AS and a link represents an inter-domain connection
between two neighboring ASes. Links are labeled with the business relationship
linking the two ASes (customer, peer or provider). We computed the actual
forwarding paths following the routing tree algorithm described
in~\cite{goldberg_how_secure_are_interdomain_routing_protocols} which takes
into account the business relationships between ASes.

\remove{
We build a model of the Internet topology by processing more than \emph{2.5 million} BGP routes (covering all Internet prefixes) along with \emph{4 billion} BGP updates advertised on 182 BGP sessions maintained by 3 RIPE BGP collectors~\cite{ripe:ris} (rrc00, rrc01 and rrc03). We collected these routes and updates daily, for 4 months, between Oct 2015 and Jan 2016.
}


\myitem{Mapping Bitcoin nodes to ASes} We finally inferred the most-specific
prefix and the AS hosting each Bitcoin node by processing more than \emph{2.5
million} BGP routes (covering all Internet prefixes) advertised on 182 BGP
sessions maintained by 3 RIPE BGP collectors~\cite{ripe:ris} (rrc00, rrc01 and
rrc03). The mapping is done by associating each prefix to the origin AS
advertising it and by validating the stability of that origin AS over time
(to avoid having the mapping polluted by hijacks).


\remove{
\begin{table}[t]
\centering
\small
\def\arraystretch{1}
\begin{tabular}{@{}lclcl@{}}
\toprule
\emph{Observation} && \emph{for partition}  && \emph{for delay}   \\
\midrule
>90\% of nodes in short prefixes && Good && Good  \\
Concentration in AS hosting && Bad && Good \\
Concentration in AS paths && - && Good  \\
Destributed/Muti-homed pools && Bad && Bad  \\

\bottomrule
\end{tabular}
\caption{We summarize here the relevant network characteristics of the Bitcoin network as well as their consequences with regards to our attacks. Bad means that it makes the attack more difficult and good easier. }
\label{tab:insides}
\end{table}
}

\subsection{Findings}
\label{ssec:topo_findings}

We now discuss the key characteristics of the Bitcoin network from the Internet routing perspective. We explain which of them constitute enablers or hindrances for an AS-level attacker. 

\myitem{A few ASes host most of the Bitcoin nodes} Fig.~\ref{fig:cdf_bitcoin_per_ases} depicts the
cumulative fraction of Bitcoin nodes as a function of the number of hosting
ASes. We see that only 13 (resp. 50) ASes host 30\% (resp. 50\%) of the entire
Bitcoin network. 
These ASes pertain to broadband providers such as Comcast (US), Verizon (US) or Chinanet (CN) as well as to cloud providers such as Hetzner (DE), OVH (FR) and Amazon (US). We observe the same kind of concentration when considering the distribution of Bitcoin nodes per IP prefix: only 63 prefixes
(0.012\% of the Internet) host 20\% of the network.

Regarding delay attacks, this high concentration makes Bitcoin traffic more
easy to intercept and therefore more vulnerable. With few ASes hosting many
nodes, any AS on-path (including the host ASes) is likely to intercept many
connections at once, making delay attacks more disruptive. Regarding partition
attacks, the effect of the concentration is a bit more nuanced. Indeed, high concentration reduces the total number of feasible partitions because of intra-AS connections that cannot be intercepted (Section~\ref{ssec:hidden}). At the same time, tough, the remaining feasible partitions are much easier to achieve since they require fewer hijacked prefixes (Section~\ref{sec:partition}).


\begin{figure}[t]
 \centering
 \begin{subfigure}[t]{0.47\columnwidth}
\includegraphics[width=.9\columnwidth]{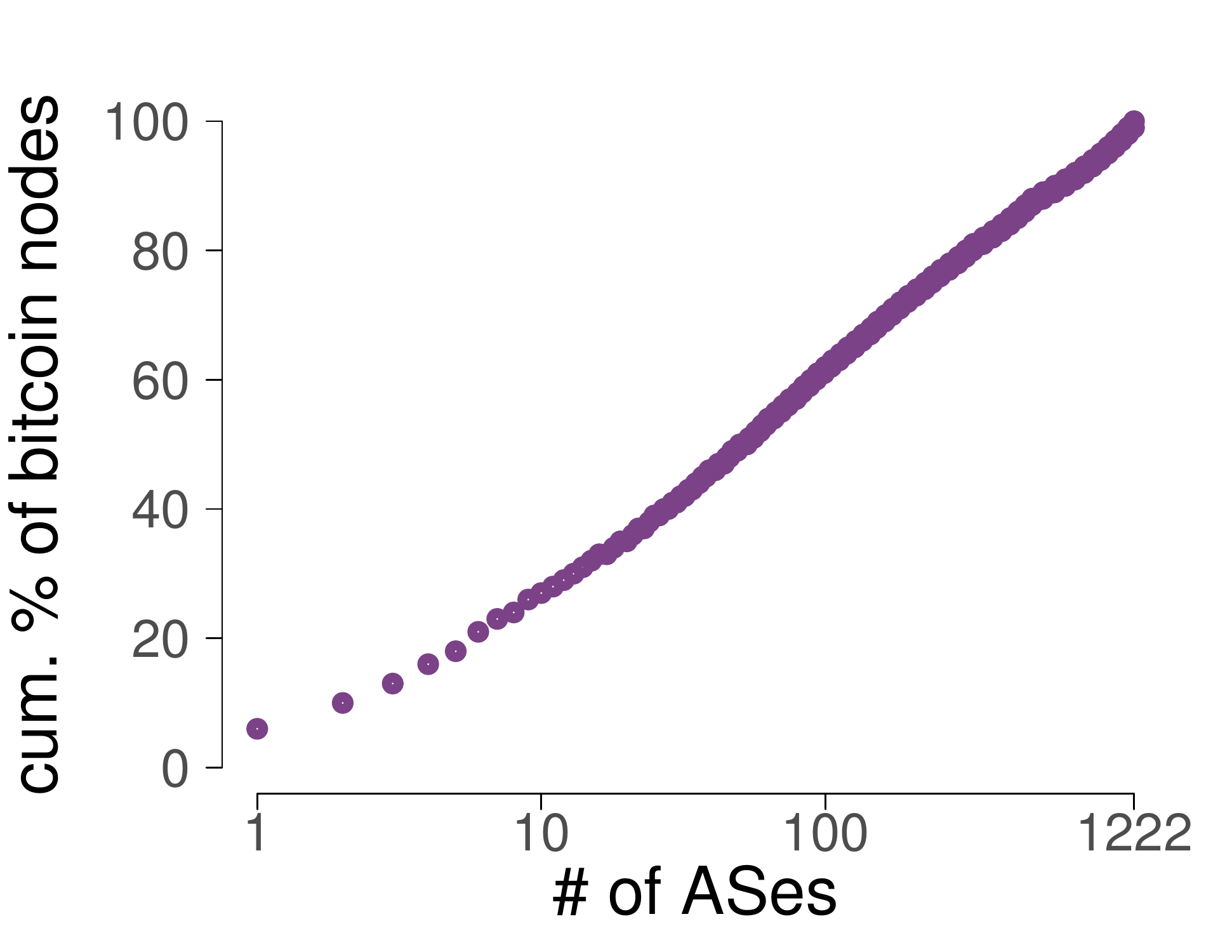}
	 \caption{Only
13 ASes host 30\% of the entire network, while 50 ASes host 50\% of the Bitcoin
network.}
	 \label{fig:cdf_bitcoin_per_ases}
 \end{subfigure}
 \hfill
 \begin{subfigure}[t]{0.47\columnwidth} 	\includegraphics[width=.9\columnwidth]{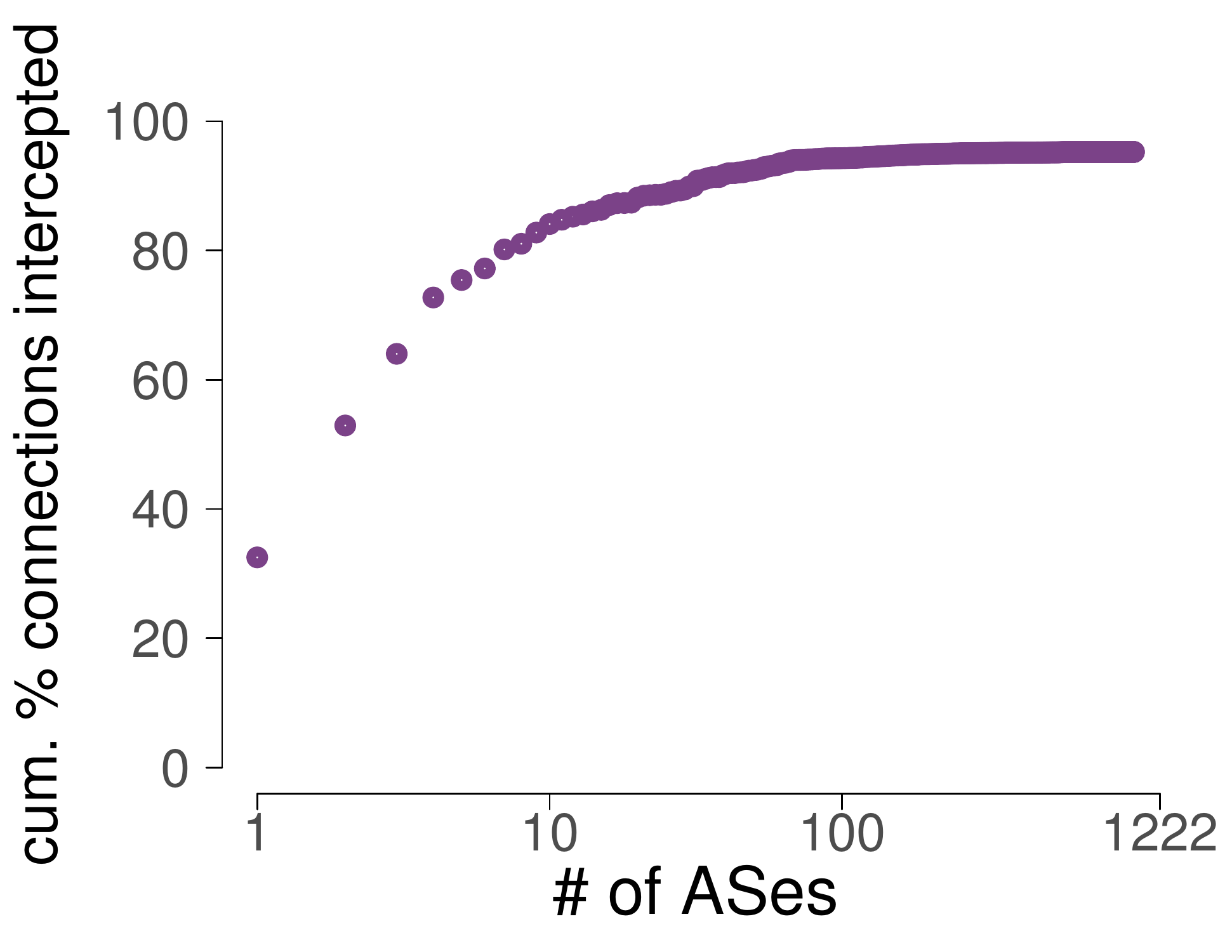}
	 \caption{Few ASes intercept large percentages of Bitcoin traffic: 3 of them intercept 60\% of all possible Bitcoin connections.}
	 \label{fig:cdf_connection_per_ases}
 \end{subfigure}
 \caption{Bitcoin is heavily centralized from a routing viewpoint.}
 \label{fig:centralization}
\end{figure}

\myitem{A few ASes naturally intercept the majority of the Bitcoin traffic}
Large transit providers (\emph{i.e.}, Tier-1s) tend to be traversed by a large fraction of all the Bitcoin connections.
Fig.~\ref{fig:cdf_connection_per_ases} depicts the cumulative percentage of
connections that can be intercepted by an increasing number of ASes
(\emph{e.g.}, by colluding with each other). We see that only \emph{three}
ASes, namely Hurricane Electric, Level3, and Telianet, can together intercept
\emph{more than 60\% of all possible Bitcoin connections}, with Hurricane alone
being on path for 32\% of all connections. 

Regarding delay attacks, these few ASes could act as powerful delay attackers.
Regarding partition attacks, this observation does not have any direct
implication as partitioning requires a \emph{full} cut to be effective
(Section~\ref{sec:partition}).


\myitem{$>$90\% of Bitcoin nodes are vulnerable to BGP hijacks} 93\% of all
prefixes hosting Bitcoin nodes are shorter than /24, making them vulnerable to
a global IP hijack using more-specific announcements. Indeed, prefixes strictly
longer than /24 (i.e., /25 or longer) are filtered by default by many ISPs. Observe that the remaining 7\% hosted in /24s are \emph{not}
necessarily safe. These can still be hijacked by another AS performing a
shortest-path attack, i.e., the attacker, who will advertise a /24 just like
the victim's provider will attract traffic from all ASes that are closer to her in terms of number of hops.

While this finding does not have a direct impact on delay attacks, it clearly helps partition attackers as they can divert almost all Bitcoin traffic to their infrastructure (modulo stealth connections, see Section~\ref{sec:partition}).

\myitem{Mining pools tend to be distributed and multi-homed}
Mining pools have a complex infrastructure compared to regular nodes.
We found that \emph{all} pools use at least two ASes to connect to the
Bitcoin network, while larger pools such as Antpool, F2Pools,
GHash.IO, Kano use up to 5 ASes. 


Pool multi-homing makes both network attacks more challenging and is one of the
main precaution measures node owners can use against routing attacks. While harder, routing attacks are still possible in the presence of multi-homing as
we illustrate in Section~\ref{sec:delay_results}.

\myitem{Bitcoin routing properties are stable over time} While numerous nodes
continuously join and leave the Bitcoin network, the routing properties
highlighted in this section are stable. As validation, we ran our analysis
daily over a 4 month period. We found that the same IPs were present on average
for 15.2 consecutive days (excluding IPs that were seen only once). Moreover, 50 ASes hosted each day 49.5\% of Bitcoin clients (standard deviation: 1.2\%)
while 24.7\% of Bitcoin nodes are found daily in just 100 prefixes (standard
deviation: 1.77\%).

\section{Partitioning Bitcoin: Evaluation}
\label{sec:eval_partition}

In this section, we evaluate the practicality and effectiveness of partitioning
attacks by considering four different aspects of the attack. \emph{First}, we show that
diverting Bitcoin traffic using BGP hijacks works in practice by performing
an actual hijack targeting our own Bitcoin nodes (Section~\ref{ssec:hijackExp}). \emph{Second}, we show that
hijacking fewer than 100 prefixes is enough to isolate a large
amount of the mining power due to Bitcoin's centralization (Section~\ref{ssec:nb_prefixes_hijack}). \emph{Third}, we show
that much larger hijacks already happen in the Internet today, some already
diverting Bitcoin traffic (Section~\ref{ssec:hijack_real}). \emph{Fourth}, we show that Bitcoin quickly recovers
from a partition attack once it has stopped (Section~\ref{ssec:methodology_experiment}).

\subsection{How long does it take to divert traffic with a hijack?}
\label{ssec:hijackExp}

We hijacked and isolated our own Bitcoin nodes which were connected to the
live network via our own public IP prefixes. In the following, we
describe our methodology as well as our findings with regard to the time it
takes for a partition to be established.

\myitem{Methodology} We built our own virtual AS with full BGP connectivity
using Transit Portal (TP)~\cite{peering}. TP provides virtual ASes with BGP
connectivity to the rest of the Internet by proxying their BGP announcements
via multiple worldwide deployments, essentially acting as a multi-homed provider to 
the virtual AS. In our experiment, we used the \texttt{Amsterdam} TP deployment as
provider and advertised 184.164.232.0/22 to it. Our virtual AS hosted six bitcoin 
nodes (v/Satoshi:0.13.0/). Each node had a public IP in 184.164.232.0/22 (.1 to .6 addresses)
 and could therefore accept connections from any other Bitcoin node in the Internet.

\remove{
traffic destined to IP addresses within that range will be routed first to the
Amsterdam01 TP, and then sent to our machines via the corresponding tunnel.
}

We performed a partition attack against our 6 nodes using BGP hijacking. For
this, we used \texttt{Cornell}, another TP deployment, as the malicious AS.
Specifically, we advertised the prefix 184.164.235.0/24 via the
\texttt{Cornell} TP. This advertisement is a more-specific prefix with respect
to the announcement coming from the \texttt{Amsterdam} TP and covers
all the IPs of our nodes. Thus, after the new prefix announcement is propagated
throughout the Internet, Bitcoin traffic directed to any of our nodes will
transit via \texttt{Cornell} instead of \texttt{Amsterdam}. To mimic an
interception attack (Section~\ref{sec:background}), we configured the
\texttt{Cornell} TP to forward traffic back to our AS. As such,
connections to our nodes stayed up during the hijack even though they
experienced a higher delay.

We performed the attacks 30 times and measured the time elapsed
from announcement of the most specific prefix until all traffic towards our
nodes was sent via the \texttt{Cornell} TP.

\begin{figure}[t]
 \centering
 \includegraphics[width=.6\columnwidth]{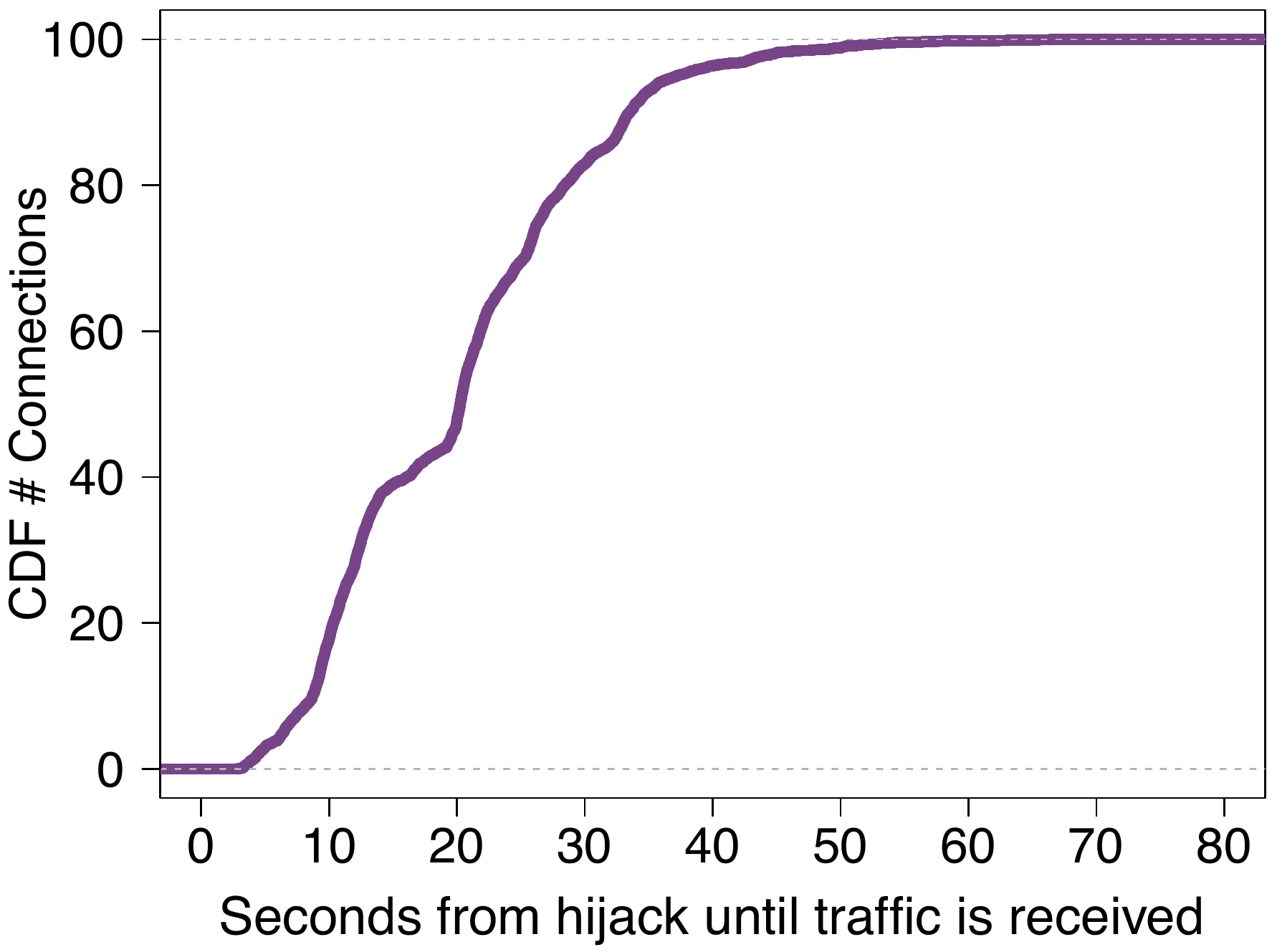}
 \caption{Intercepting Bitcoin traffic using BGP hijack is fast and effective: all the traffic was flowing through the attacker within 90 seconds. Results computed while performing an \emph{actual} BGP hijack against our own Bitcoin nodes.}
 \label{fig:hijack}
\end{figure}

\myitem{Diverting Bitcoin traffic via BGP is fast (takes <2 minutes)} The
results of our experiment are shown in Fig.~\ref{fig:hijack}. The main insight
is that the attacker received the hijacked traffic very quickly. After only 20 seconds, more than 50\% of the connections are diverted. Within 1.5 minutes, all traffic was transiting via the malicious AS. Thus, attacked nodes are effectively isolated almost as soon as the hijack starts.

We took great care to ensure that our experiments did not negatively impact the actual Bitcoin
network. We discuss the ethical considerations behind our experiments in
Appendix~\ref{sssec:ethic}.

\remove{More importantly this number should be compared to the time it will take to stop the attacker by disconnecting the attacker from the Internet which will take at least a couple of hours.}  

\remove{
\myitem{Practicality} 
Such an attack can be performed against the vast majority of the Bitcoin nodes in the current network. Particularly, to perform a partition attack the adversary needs to intercept all traffic from hijacked set to the rest of the network.  nodes are those which are hosted in AS with gateways as well as in ASes with nodes which use /24. That is because if we assume that a potential victim connects to a node in the same AS but which use /24 or is a gateway then the attacker will not intercept this traffic. Even if we consider these nodes safe 61\% of nodes are still vulnerable. Particularly, 92\% use a /24 prefix, 82.6\% have no /24 in their AS and 76\% have no gateway in their AS. It is very important to note here that the rest 40\% is also vulnerable to be affected by a partition attack. The difference is that for these nodes can be placed in any of the two components.
}


\subsection{How many prefixes must be hijacked to isolate mining power?}
\label{ssec:nb_prefixes_hijack}

Having shown that hijacking prefixes is an efficient way to divert Bitcoin
traffic, we now study the practicality of isolating a specific set of nodes $P$.
We focus on isolating sets holding
mining power because they are: \emph{(i)} more challenging to perform (as
mining pools tend to be multi-homed); and \emph{(ii)} more disruptive as
successfully partitioning mining power can lead to the creation of parallel
branches in the blockchain.

To that end, we first
estimate the number of prefixes the attacker needs to hijack to isolate a specific set of nodes as a function of the mining power they hold. In the following subsection, we evaluate how practical
a hijack of that many prefixes is with respect to the hijacks that frequently take place in the Internet. 

\remove{we evaluate how practical
is a hijack of that many prefixes with respect to the hijacks that frequently take place in the Internet. }

\myitem{Methodology} As described in Section~\ref{sec:partition}, not all sets
of nodes can be isolated as some connections cannot be intercepted. We
therefore only determine the number of prefixes required to isolate sets $P$ that are
feasible in the topology we inferred in Section~\ref{sec:topology}. In
particular, we only consider the sets of nodes that contain: \emph{(i)} either
all nodes of an AS or none of them; and \emph{(ii)} either the entire mining pool, namely all of its 
gateways or none of them. With these restrictions, we essentially avoid the possibility of having any leakage point within $P$, that is caused by an intra-AS or intra-pool stealth connection. 
However, we cannot account for secret peering agreements that may or may not exist between pools. Such agreements are inherently kept private and their existence is difficult to ascertain.

\remove{
6\%\-94\% && 2 && 86.0 && 20\\
7\%\-93\% && 7 && 72.0 && 23\\
8\%\-92\% && 32 && 69.5 && 14\\
30\%\-70\% && 83 && 83.0 && 1\\
39\%\-61\% && 32 && 51.0 && 11\\
40\%\-60\% && 37 && 80.0 && 8\\
41\%\-59\% && 44 && 55.0 && 3\\
45\%\-55\% && 34 && 41.0 && 5\\
46\%\-54\% && 78 && 78.0 && 1\\
47\%\-53\% && 39 && 39.0 && 1\\
}

\begin{table}[t]
\centering
\small
\def\arraystretch{1.1}
\setlength{\tabcolsep}{7.5pt}
\begin{tabular}{l l l l}
\toprule
\begin{tabular}{@{}l@{}} \emph{Isolated} \\ \emph{mining power}\end{tabular} & \begin{tabular}{@{}l@{}} \emph{min. \# pfxes} \\ \emph{to hijack} \end{tabular} & \begin{tabular}{@{}l@{}} \emph{median \# pfxes} \\ \emph{to hijack}\end{tabular} & \begin{tabular}{@{}l@{}} \emph{\# feasible} \\ \emph{partitions}\end{tabular}   \\

\midrule
8\%\ & 32 & 70 & 14\\
30\%\ & 83 & 83 & 1\\
40\%\ & 37 & 80 & 8\\
47\%\ & 39 & 39 & 1\\
\bottomrule
\end{tabular}
\caption{Hijacking <100 prefixes is enough to feasibly partition {\raise.17ex\hbox{$\scriptstyle\sim$}}50\% of the mining power. Complete table in Appendix~\ref{ssec:partitions}.}
\label{tab:partition_costs}
\end{table}

\remove{ 
\myitem{Hijacking <100 prefixes is enough to isolate {\raise.17ex\hbox{$\scriptstyle\sim$}}50\% of Bitcoin mining power}
As described in Section~\ref{sec:partition}, not all sets of nodes can be
isolated as some connections cannot be intercepted. We therefore evaluate the
number of prefixes required to isolate sets $P$ that are feasible in the
topology we inferred in Section~\ref{sec:topology}. In particular, we only
consider the sets of nodes that contain: \emph{(i)} either all nodes of an AS
or none of them; and \emph{(ii)} either all gateways of a mining pool or none
of them. These restrictions restricts the values of the isolated mining power
to few discrete values (since not all partitions can be realized). Moreover,
some pairs of pools cannot be disconnected as they have gateways within the
same AS.

Particularly, an attacker can isolate a set of nodes $P$, if there are no stealth connections from nodes inside $P$ to nodes outside it. We explained in Section~\ref{sec:partition} how the attacker can prevent such connections from bridging the isolated set to the rest of the network by monitoring the connections she hijacks.
For the purpose of off-line estimating the number of required prefix hijacks though,\maria{with off-line I mean for estimating-inferring the number that would be needed as opposed to performing the attack and decide on the fly}
}

\myitem{Hijacking <100 prefixes is enough to isolate {\raise.17ex\hbox{$\scriptstyle\sim$}}50\% of Bitcoin mining power}
In Table~\ref{tab:partition_costs} we show the number of different feasible sets of nodes we found containing the same amount of mining power (4th and 1st column, respectively). We also include the minimum and median number of the prefixes the attacker would need to hijack to isolate each portion of mining power (2nd and 3rd column, respectively).
\remove{ Our results are summarized in Table I. The leftmost column describes the portion of mining power contained in the isolated set of nodes. We indicate the number of feasible sets containing the same amount of mining power along with the minimum and median amount of prefixes to hijack to create the partitions in practice.
}

As predicted by the centralization of the Bitcoin network
(Section~\ref{sec:topology}), the number of prefixes an attacker needs to hijack to create a
feasible partition is small: hijacking less than 100 prefixes is enough to
isolate up to 47\% of the mining power. As we will describe next, hijack events
involving similar numbers of prefixes happen regularly in the Internet. Notice
that the number of prefixes is not proportional to the isolated mining power.
For example, there is a set of nodes representing 47\% of mining power that can
be isolated by hijacking 39 prefixes, while isolating 30\% of the mining power belonging to different pools would require 83 prefixes. Indeed, an attacker can isolate additional mining power
with the same number of hijacked prefixes when several pools
in the isolated set are hosted in the same ASes.

\remove{
While our results are inherently dependent on the accuracy of our inferred
topology (Section~\ref{sec:topology}), we stress that attackers may slowly gain
better information, e.g., through repeated stealthy hijacks of the target. \maria{I don't understand that}
}

\begin{figure}[t]
	\centering
	\begin{subfigure}[t]{0.47\columnwidth}
		\includegraphics[width=\columnwidth]{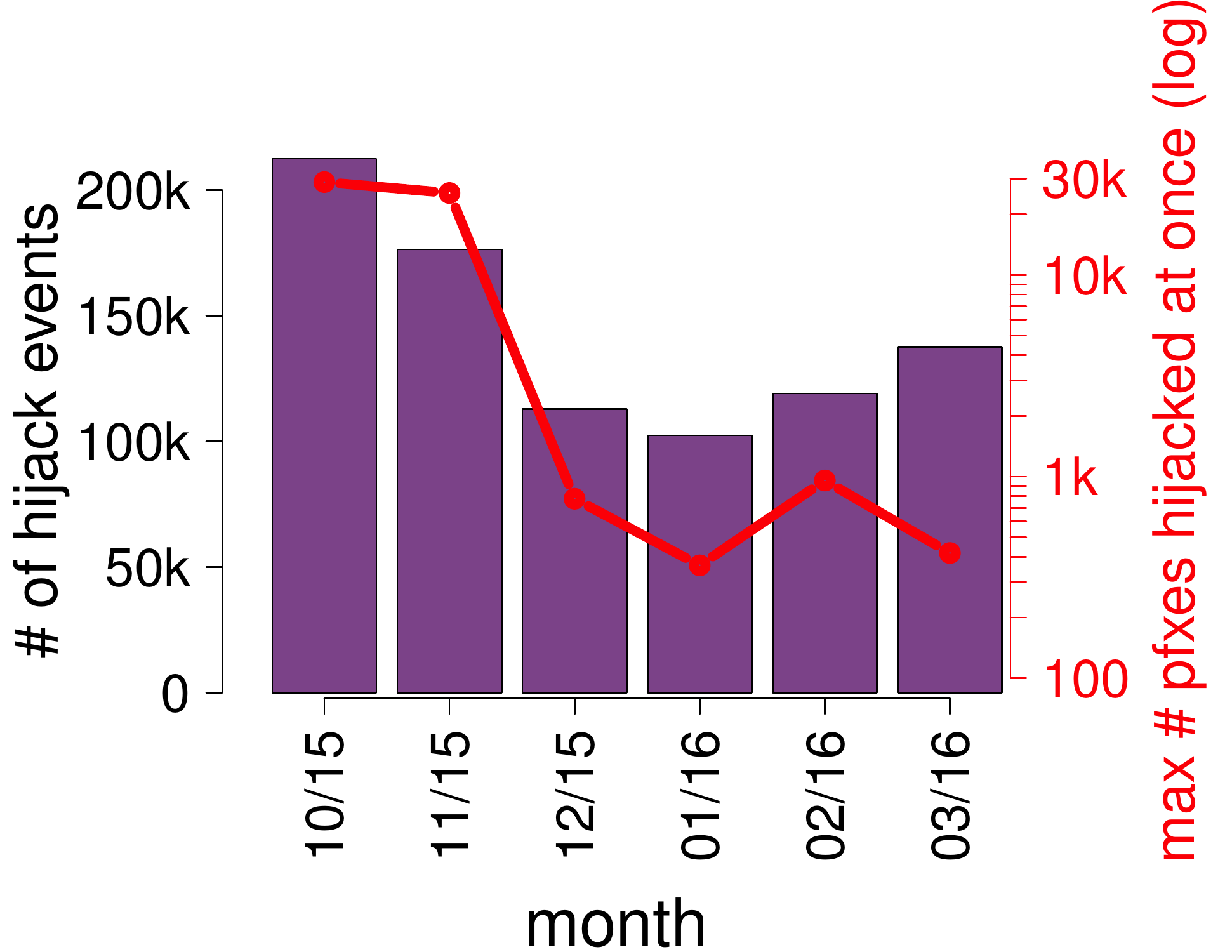}
		\caption{Each month sees \emph{at least} 100{,}000 hijacks, some of which involve \emph{thousands} of prefixes.}
		\label{fig:hijack_number}
	\end{subfigure}
	\hfill
	\begin{subfigure}[t]{0.47\columnwidth} 	\includegraphics[width=\columnwidth]{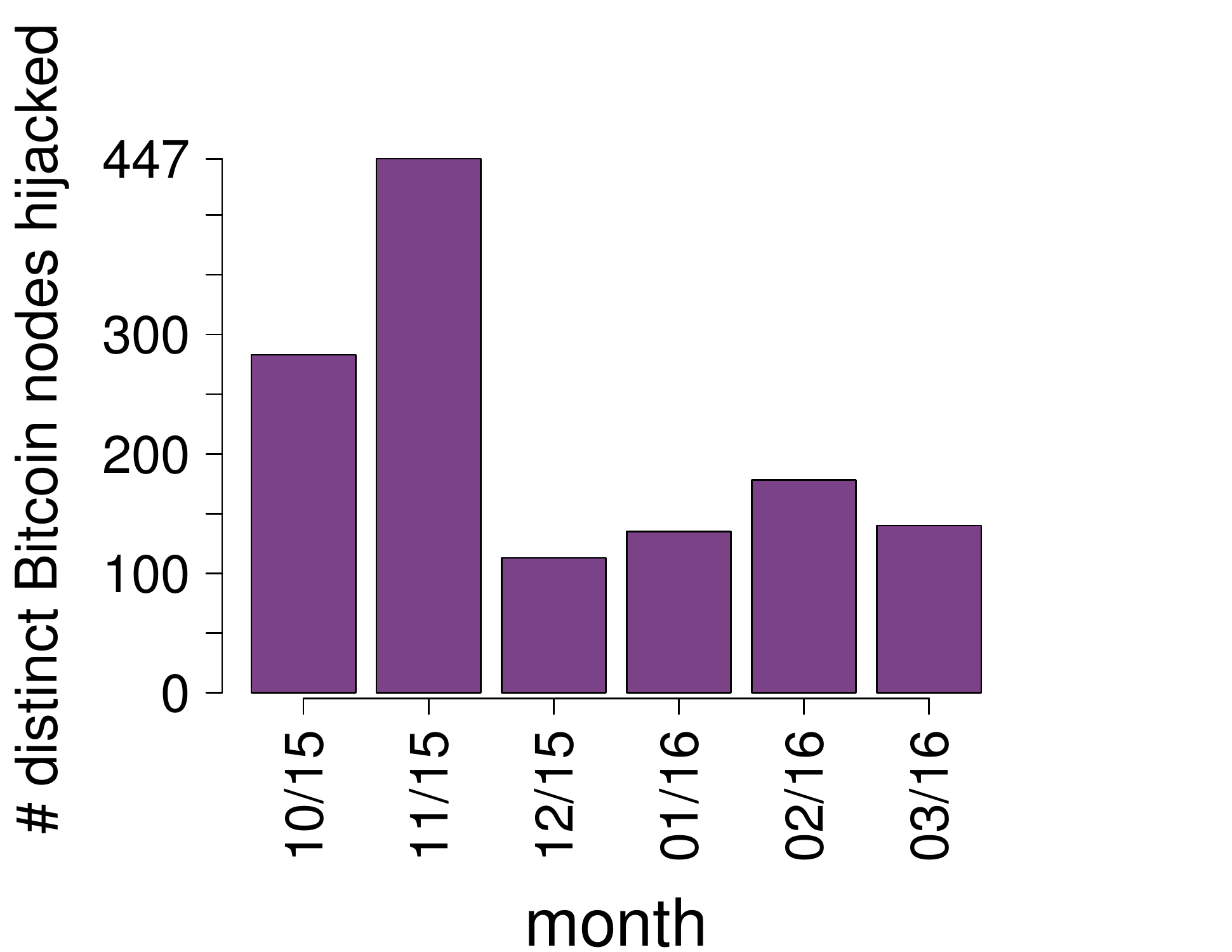}
		\caption{Each month, traffic for at least 100 \emph{distinct} Bitcoin nodes end up diverted by hijacks.}
		\label{fig:hijack_size}
	\end{subfigure}
	\caption{Routing manipulation (BGP hijacks) are prevalent today and do impact Bitcoin traffic.}
	\label{fig:hijack_analysis}
\end{figure}

\vfill
\subsection{How many hijacks happen today? Do they impact Bitcoin?}
\label{ssec:hijack_real}

Having an estimate of the number of prefixes that need to be hijacked to partition the entire network, we now look at how common such hijacks
are over a 6-months window, from October 2015 to March 2016. We show that BGP
hijacks are not only prevalent, but also end up diverting Bitcoin traffic.

\newpage
\noindent\textbf{Methodology} We detected BGP hijacks by processing \emph{4 billion} BGP
updates advertised during this period on 182 BGP sessions maintained by 3 RIPE
BGP collectors~\cite{ripe:ris} (rrc00, rrc01 and rrc03). We consider an update
for a prefix $p$ as a hijack if the origin AS differs from the origin AS
seen during the previous month. To avoid false positives, we do not consider
prefixes which have seen multiple origin ASes during the previous month. We
count only a single hijack per prefix and origin pair each day: if AS $X$
hijacks the prefix $p$ twice in one day, we consider both as part of a single
hijack.

\myitem{Large BGP hijacks are frequent in today's Internet, and already end up
diverting Bitcoin traffic}
Fig.~\ref{fig:hijack_analysis} summarizes our results. We see that there are \emph{hundreds of thousands} of hijack events \emph{each
month} (Fig.~\ref{fig:hijack_number}). While most of these hijacks involve a single IP prefix, large hijacks
involving between 300 and 30,000 prefixes are also seen every month (right axis).
Fig~\ref{fig:hijack_size} depicts the number of Bitcoin nodes for which traffic was diverted in these hijacks. Each month, at least 100 Bitcoin
nodes are victim of hijacks\footnote{The actual hijack attempt may have been aimed at other services in the same IP range, still, these nodes were affected and their traffic was re-routed.}. As an illustration, 447 distinct nodes
({\raise.17ex\hbox{$\scriptstyle\sim$}}7.8\% of the Bitcoin network)
ended up hijacked at least once in November 2015. 

\remove{
While frequent, BGP hijacks are hard to diagnose and slow to resolve. As BGP
securities extensions (BGPSec~\cite{bgpsec}, RPKI~\cite{rfc6480}) are seldom
deployed~\cite{lychev_bgpsec_deployment}, operators rely on simple monitoring
systems that dynamically analyze and report rogue BGP announcements
(e.g.,~\cite{yan2009bgpmon,chi2008cyclops}). 

In Section~\ref{sec:bitcoin_interception}, we showed that any single AS can
intercept 50\% of the connections (and mining power) by hijacking few prefixes
({\raise.17ex\hbox{$\scriptstyle\sim$}}0.15\% of all prefixes). We now
investigate the influence an attacker might have by simply dropping all
messages on the intercepted connections, effectively partitioning the Bitcoin
network.

To measure the effect of a partition in a large-scale virtualized environment
composed of 1050 actual Bitcoin nodes
(\S\ref{ssec:methodology_experiment}). We describe two key findings. First, the network
quickly goes down after connections are disrupted. Second, while connections
crossing the partition slowly return as the partition heals (\S\ref{sec:partition_results}), the fork rate is quickly restored to normal. The latter result implies that
imperfect partitions are not useful in disrupting the network, which makes
partition attacks last only as long as the attacker successfully isolates the
different parts of the network.
}

\begin{figure}[t]
 \centering
 \includegraphics[width=.65\columnwidth]{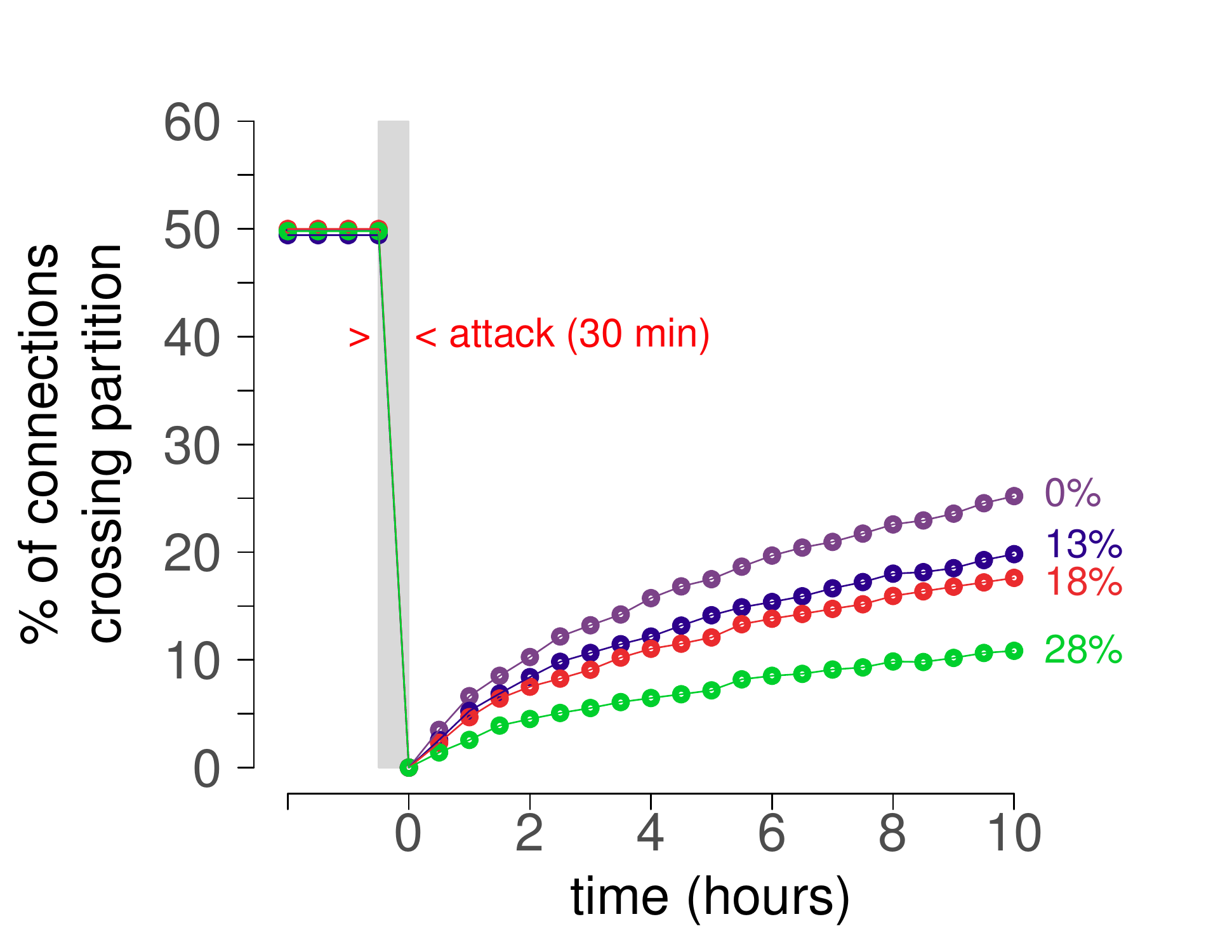}
 \caption{Bitcoin heals slowly after large partition attacks. After 10h, only half as many connections cross the partition. Healing is even slower if the attacker is naturally on-path for 13\%, 18\%, 28\% of the connections.}
 \label{fig:partition_healing}
 \vspace{-.5cm}
\end{figure}

\subsection{How long do the effects of a partition attack last?}
\label{ssec:methodology_experiment}
Having investigated the methodology and the relative cost of creating a
partition, we now explore how quickly the Bitcoin network recovers from a partition attack. We found out that while the two components quickly reconnect, they stay sparsely connected for a very long time. We first describe the
experimental set-up. Next, we explain why partitions are not persistent in
practice and briefly hint on how an attacker could prolong their lifetime.

\myitem{Methodology} We build a testbed composed of 1050 Bitcoin clients
running the default \verb+bitcoind+ core (v0.12.1) in testnet mode. Each node
runs in a virtual machine connected to a virtual switch and is configured with
a different random IP address. Nodes automatically connect to other nodes in the testbed. We enforced a 50\%--50\% partition, by installing
drop rules on the switch which discard any packet belonging to a connection
crossing the partition. Observe that a 50\%--50\% split is the easiest partition to recover from, as
after the attack the chance that a new connection would bridge the two halves is maximal. We measure the partition recovery time by recording
the percentage of connections going from one side to the other in 30 minute intervals.

Bitcoin TCP connections are kept alive for extended periods. As
such, new connections are mostly formed when nodes reconnect or leave the
network (churn). To simulate churn realistically, we collected the lists of all
reachable Bitcoin nodes~\cite{bitnodes}, every 5 minutes, from February 10 to
March 10 2016. For every node $i$ connected in the network on the first day, we
measured $t_i$ as the elapsed time until its first disappearance. To determine
the probability of a node to reboot, we randomly associated every node in our
testbed with a type $t_i$ and assumed this node reboots after a period of time
determined by an exponential distribution with parameter
$\lambda=\frac{1}{t_i}$. The time for next reboot is again drawn according to
the same distribution. This method produces churn with statistics matching
the empirical ones. We repeat each measurement at least 5 times and report the median value found.

\myitem{Bitcoin quickly recovers from a partition attack, yet it takes hours for the network to be densely connected again}
We measured how long it takes for the partition to heal by measuring how
many connections cross it, before, during and after its formation (Fig.~\ref{fig:partition_healing}). We consider
two different attack scenarios: \emph{(i)} the adversary does not intercept any
bitcoin traffic before or after the attack; and \emph{(ii)} the adversary intercepts some connections naturally.

\emph{Case 1: The adversary intercepts no traffic after the attack.} It takes
2 hours until one fifth of the initial number of connections crossing the
partition are established, while after 10 hours only half of the connections
have been re-established. The slow recovery is due
to the fact that nodes on both sides do not actively change their connections
unless they or their neighbors disconnect.

\emph{Case 2: The adversary intercepts some traffic after the attack.} If an
AS-level adversary is naturally on-path for some of the connections, she can
significantly prolong the partition's lifetime. To do so, the attacker would
just continue to drop packets on connections she naturally intercepts. We measured
the effect of such attacks for attackers that are on-path for 14\%, 18\%, and
28\% of the connections, respectively (Fig.~\ref{fig:partition_healing}). We
see that an AS-adversary who is initially on-path for 28\% of the connections
can prolong the already slow recovery of the partition by 58\%. Many other ways to increase the persistence of a partition exist. Due to space constraints, we discuss some of them in Appendix~\ref{ssec:partition_life}.

Despite the long healing time, the orphan rate of the network returned to
normal even with 1\% of all connections crossing the partition. This fact shows
that partitions need to be perfect in order to affect the network significantly.

\section{Delaying Block Propagation: Evaluation}
\label{sec:delay_results}

In this section, we evaluate the impact and practicality of delay attacks
through a set of experiments both at the node-level and at the network-wide
level. We start by demonstrating that delay attacks against a single node work
in practice by implementing a working prototype of an interception software that we then use to delay our own Bitcoin nodes
(Section~\ref{ssec:network_single_node_attack}). We then evaluate the impact of
network-wide delay attacks by implementing a scalable event-driven Bitcoin simulator. In constrast to partitioning attacks, and to targeted delay attacks, we show that Bitcoin is well-protected against network-wide delay attacks, even
when considering large coalitions of ASes as attackers (Section~\ref{ssec:network_delay_attack}).

\subsection{How severely can an attacker delay a single node?}
\label{ssec:network_single_node_attack}


\myitem{Methodology} 
We implemented a prototype of our interception software on top of Scapy~\cite{scapy}, a
Python-based packet manipulation library. Our prototype follows the methodology of Section~\ref{sec:delay} and is efficient both
in terms of state maintained and processing time. Our
implementation indeed maintains only 32B of memory (hash size) for each peer of the
victim node. Regarding processing time, our implementation leverages pre-recompiled
regular expressions to efficiently identify critical packets (such as those
with \textsf{BLOCK} messages) which are then
processed in parallel by dedicated threads. Observe that the primitives required for the interception software are also supported by recent programmable data-planes~\cite{bosshart2014p4} opening up the possibility of performing the attack entirely on network devices.



We used our prototype to attack one of our own Bitcoin nodes
(v/Satoshi:0.12.0/, running on Ubuntu 14.04). The prototype ran on a machine acting as a gateway to the
victim node. Using this setup, we measured the effectiveness of an attacker in
delaying the delivery of  blocks, by varying the percentage of connections she intercepted. To that end, we measured the fraction of time during which the victim was uninformed of the most recently mined block.
We considered our victim node to be uninformed when its
view of the main chain is shorter than that of a reference node. The reference
node was another Bitcoin client running the same software and the same number of peers as the victim, but without any attacker.


\newpage
\noindent\textbf{Delay attackers intercepting 50\% of a node's connection can waste 63\% of its mining power} 
Table~\ref{tab:node_results} illustrates the percentage of the victim's uptime, during which it was uniformed of the last mined block, considering that the attacker intercepts 100\%, 80\%, and 50\% of its connections.
Each value is the average over an attack period of {\raise.17ex\hbox{$\scriptstyle\sim$}}\emph{200 hours}. To further evaluate the practicality of the attack, the table also depicts the fraction of Bitcoin nodes for which there is an AS, \emph{in addition to their direct provider}, that intercepts 100\%, 80\%, and 50\% of its connections.

Our results reflect the major strength of the attack, which is its
effectiveness even when the adversary intercepts only a fraction of the
victim's connections. Particularly, we see that an attacker can waste 63\% of a
node's mining power by intercepting half of its connections. Observe that, even
when the attacker is intercepting all of the victim's connections, the victim eventually
gets each block after a delay of 20 minutes.

Regarding the amount of vulnerable nodes to this attack in the Bitcoin
topology, we found that, for 67.9\% of the nodes, there is at least one AS
\emph{other than their provider} that intercepts more than 50\% of their
connections. For 21.7\% of the nodes there is in fact an AS (other than their provider) that intercepts all their connections to other nodes. In short, 21.7\% of the nodes can be isolated by an AS that is not even
their provider.

\begin{table}[t]

\centering
\small
\def\arraystretch{1.1}
\setlength{\tabcolsep}{3pt}
\begin{tabular}{@{}lclclcl@{}}

\toprule
\centering

\% intercepted connections  && 50\% && 80\%  && 100\% \\
\midrule
\% time victim node is uniformed	&& 63.21\% 	&& 81.38\%  && 85.45\% 	\\
\% total vulnerable Bitcoin nodes && 67.9\%  &&  38.9\%  &&  21.7\%  \\
\bottomrule
\end{tabular}
\caption{67.9\% of Bitcoin nodes are vulnerable to an interception of 50\% of their connections by an AS \emph{other than their direct provider}.
Such interception can cause the node to lag behind a reference node 63.21\% of the time.}
\label{tab:node_results}
\end{table}

\subsection{Can powerful attackers delay the entire Bitcoin network?}
\label{ssec:network_delay_attack}

Having shown that delay attacks against a single node are practical, we now quantify the network-wide effects of delaying block propagation. 

Unlike partitioning attacks, we show that network-wide delay attacks (that do not utilize active hijacking) are unlikely to happen in practice. Indeed, only extremely powerful attackers such
as a coalition grouping all ASes based in the US could significantly delay the
Bitcoin network as a whole, increasing the orphan rate to more than 30\% from
the normal 1\%. We also investigate how this effect changes as a function of the degree of multi-homing that pools adopt.

\myitem{Methodology} Unlike partition attacks, the impact of delay attacks on the network is difficult to ascertain. One would need to actually slow down the network
to fully evaluate the cascading effect of delaying blocks. We therefore built a
realistic event-driven simulator following the principles
in~\cite{neudecker2015simulation} and used it to evaluate such effects.

\newcommand{\specialcell}[2][c]{%
	\begin{tabular}[#1]{@{}c@{}}#2\end{tabular}}

\remove{
\begin{table*}
	\centering
	\begin{tabular}{|c|c|c|c|c|}
		\hline
		\specialcell{Monitored\\connection} & The attack & \specialcell{Preserve\\TCP} & \specialcell{Preserve\\BTC} & Outcome \\
		\hline
		any & drop tcp packets & no & no & connection lost after tcp timeout\\
		\hline
		from node & \specialcell{change block hash in getdata\\ message} & yes & no & bitcoind drops connection after 20 min\\
		\hline
		from node & \specialcell{change block hash in getdata\\ message, inject again within 20 min} & yes & yes & \specialcell{block delayed 20 min \\ (connetion kept)}\\
		\hline
		to node & make block header invalid & yes & no & bitcoind drops\left(  connection after 20 min\\
		\hline
	\end{tabular}
	\caption{The effect of different attacks}
\end{table*}
}

Our simulator models the entire Bitcoin network and the impact of a delay
attack considering the worst-case scenario for the attacker. Specifically, it
assumes that the communication between gateways of the same pool cannot be
intercepted. Also, pools act as relay networks in that they propagate all blocks that
they receive via all their gateways. Moreover, the simulator assumes that the
delay attacker is only effective if she intercepts the traffic \emph{from} a node that receives a block (essentially if she is able to perform the
attack depicted in Fig.~\ref{fig:delay_attack_left}). 
We provide further details on our simulator and how we evaluated it in Appendix~\ref{ssec:simulation}.

We ran our simulator on realistic topologies as well as on synthetic ones with
higher or lower degrees of multi-homing. The realistic topology was inferred as
described in Section~\ref{sec:topology}. The synthetic ones were created by
adding more gateways to the pools in the realistic topology until all pools
reached the predefined degree of multi-homing.


\begin{table}[t]
\centering
\small
\def\arraystretch{1.1}
\setlength{\tabcolsep}{6.5pt}
\begin{tabular}{cccccc}

\toprule
\multirow{2}{*}{\textit{Coalition}} & \multirow{2}{*}{\textit{Realistic topology}} & \multicolumn{4}{c}{\textit{Multihoming degree of pools}} \\
\cmidrule{3-6}
& (Section~\ref{sec:topology})	& \textit{1}		& \textit{3}		& \textit{5} & \textit{7} \\
\midrule
US & 23.78 & 38.46 & 18.18 & 6.29 & 4.20 \\
DE & 4.20 & 18.88 & 2.10 & 1.40 & 1.40 \\
CN & 4.90 & 34.27 & 1.40 & 0.70 & 0.70 \\
\bottomrule
\end{tabular}
\caption{Orphan rate (\%) achieved by different network-wide level delay attacks performed by coalitions of \emph{all} the ASes in a country, and considering either the topology inferred in Section~\ref{sec:topology} or synthetic topologies with various degrees of pool multi-homing. The normal orphan rate is \raise.17ex\hbox{$\scriptstyle\sim$}1\%.}
\label{tab:sim_results}
\vspace{-.3cm}
\end{table}

\myitem{Due to pools multi-homing, Bitcoin (as a whole) is not vulnerable to delay attackers, even powerful ones} Our results are summarized in Table~\ref{tab:sim_results}. We see that multi-homed pools considerably increase the robustness of the Bitcoin network against delay attacks. In essence, multi-homed pools act as protected relays for the whole network. Indeed, multi-homed pools have better chances of learning about blocks via at least one gateway and can also more efficiently propagate them to the rest of the network via the same gateways.

If we consider the current level of multi-homing, only powerful
attackers such as a coalition containing \emph{all} US-based ASes could effectively disrupt the network by increasing the fork rate to 23\%
(as comparison, the baseline fork rate is 1\%). In contrast, other powerful
attackers such as all China-based or all Germany-based ASes can only increase
the fork rate to 5\%. As such coalitions are unlikely to form
in practice, we conclude that \emph{network-wide} delay attacks do not constitute a threat for Bitcoin.

\myitem{Even a small degree of multi-homing is enough to protect Bitcoin from
powerful attackers.} If all mining pools were single-homed, large coalitions
could substantially harm the currency. The US for instance, could increase the
fork rate to 38\% while China and Germany could increase it to 34\% and 18\%
respectively. Yet, the fork rate drops dramatically as the average multi-homing
degree increases. This is a good news for mining pools as it shows that even a
small increase in their connectivity helps tremendously in protecting them against delay attacks.

\section{Countermeasures}
\label{sec:countermeasures}


In this section, we present a set of countermeasures against routing attacks.
We start by presenting measures that do not require any protocol change and can
be partially deployed in such a way that early adopters can benefit from higher
protection (Section~\ref{ssec:short_term}). We then describe longer-term
suggestions for both detecting and preventing routing attacks
(Section~\ref{ssec:long_term}).

\subsection{Short-term measures}
\label{ssec:short_term}



\myitem{Increase the diversity of node connections} The more connected an AS
is, the harder it is to attack it. We therefore encourage Bitcoin node owners to ensure they are multi-homed. 
Observe that even
single-homed Bitcoin nodes could benefit from extra connectivity by using one or
more VPN services through encrypted tunnels so that Bitcoin
traffic to and from the node go through multiple and distinct ASes.
Attackers that wish to deny connectivity through the tunnel would need to
either know both associated IP addresses or, alternatively, disrupt
all encrypted traffic to and from nodes---making the attack highly noticeable.

\myitem{Select Bitcoin peers while taking routing into account} Bitcoin
nodes randomly establish 8 outgoing connections. While randomness is important
to avoid biased decisions, Bitcoin nodes should establish a
few \emph{extra} connections taking routing into consideration. For this,
nodes could either issue a \verb+traceroute+ to each of their peers and
analyze how often the same AS appears in the path or, alternatively, tap into the BGP feed of their network and select their peers based on the AS-PATH. In both cases, if the same AS appears in all paths, extra random connections should be established.

\myitem{Monitor round-trip time (RTT)} The RTT towards hijacked destinations
increases during the attack. By monitoring the RTT towards its peers, a node could detect sudden changes and establish extra random connections as a protection mechanism.

\myitem{Monitor additional statistics} Nodes should deploy anomaly detection mechanisms to recognize sudden changes in: the distribution of connections, the time elapsed between a request and the corresponding answer, the simultaneous disconnections of peers, and other lower-level connection anomalies. Again, anomalies should spur the establishment of extra random connections.

\myitem{Embrace churn} Nodes should allow the natural churn of the network to
refresh their connections. A node with disabled incoming connections or even one that is behind a NAT or a firewall will never receive a random incoming
connection from the rest of the network. If the node is hijacked for a few
minutes and isolated from a part of the network, it will only reconnect to the other part upon reboot or when one of its outgoing connections fails.

\newpage
\noindent\textbf{Use gateways in different ASes} While inferring the topology we noticed that many pools were using gateways in the same AS. Hosting these gateways in different ASes would make them even more robust to routing attacks.

\myitem{Prefer peers hosted in the same AS and in /24 prefixes} As the traffic
of nodes hosted in /24 prefixes can only partially be diverted (Section~\ref{sec:background}). Hosting all
nodes in such prefixes would prevent partition attacks at the cost of
({\raise.17ex\hbox{$\scriptstyle\sim$}}1\%) increase of the total number of
Internet prefixes. Alternatively, nodes could connect to a peer hosted in a /24
prefix which belongs to their provider. By doing so they maintain a stealth
(intra-AS) connection with a node that is at least partially protected against
hijack.

\subsection{Longer-term measures}
\label{ssec:long_term}

\myitem{Encrypt Bitcoin Communication and/or adopt MAC} While encrypting Bitcoin connections
would not prevent adversaries from dropping packets, it would
prevent them from eavesdropping connections and modifying key messages. Alternatively, using a Message Authentication Code (MAC) to validate that the content of each message has not been changed would make delay attacks much more difficult. 


\myitem{Use distinct control and data channels} A key problem of Bitcoin is
that its traffic is easily identifiable by filtering on the default port
(8333). Assuming encrypted connections, the two ends could negotiate a set of random TCP ports upon connecting to each other using
the well-known port and use them to establish the actual TCP connection, on
which they will exchange Bitcoin data. This would force the AS-level adversary
to look at \emph{all the traffic}, which would be too costly. 

A simpler (but poorer) solution would be for Bitcoin clients to use randomized TCP port (encoded in clear with the \textsf{ADDR} message) as it
will force the AS-level adversary to maintain state to keep track of these ports.
Although a node can already run on a non-default port,
such a node will receive fewer incoming connections (if any) as the default client strongly prefers peers that run on the default port.

\myitem{Use UDP heartbeats} A TCP connection between two Bitcoin nodes may take
different forward and backward paths due to asymmetric routing, making AS-level
adversaries more powerful (as they only need to see one direction, see Section~\ref{sec:delay}). In addition to TCP connections, Bitcoin clients could
periodically send UDP messages with corroborating data (\emph{e.g.}, with
several recent block headers). These UDP messages can be used as a heartbeat
that will allow nodes to discover that their connection was partially
intercepted. As UDP messages do not rely on return traffic, this would
enable node to realize that they are out-of-sync and establish new connections.

\myitem{Request a block on multiple connections} Bitcoin clients could ask
multiple peers for pieces of the block. This measure would prevent misbehaving
nodes from deliberately delaying the delivery of a block, simply because in
such a case the client will only miss one fraction of the block, which it can
request from one of its other peers.


\section{Related Work}
\label{sec:related_work}


\myitem{AS-level adversaries} The concept of AS-level adversaries has been
studied before in the context of Tor~\cite{feamster:wpes2004, murdoch:pet2007,
Edman:2009:ATP:1653662.1653708,raptor2015,vanbever2014anonymity}. These works also illustrated the problems caused by centralization and
routing attacks on a distributed system running atop the Internet. Yet, Tor and
Bitcoin differ vastly in their behavior with one routing messages in an
Onion-like fashion, while the other uses \emph{random} connections to flood messages throughout the entire network.
Although random graphs are usually robust to attacks, this paper shows that
it is \emph{not} the case when the network is centralized at the routing-level.

\myitem{Bitcoin attacks} The security of Bitcoin from network-based attacks has
been relatively less explored compared to other attack scenarios. While eclipsing attacks~\cite{heilman2015eclipse,gervais2014bitcoin} have a similar impact than delay attacks when performed against a single node, they disrupt the victim's connections and assume that the attacker is directly connected to the victim (Section~\ref{sec:delay}).
For more information about Bitcoin attacks, we refer the reader to a recent comprehensive survey on the Bitcoin protocol~\cite{bonneau2015sok}.

\myitem{BGP security issues} Measuring and detecting routing attacks has seen
extensive research, both considering BGP hijacks~\cite{shi:imc12,zhang:conext07,zhang:ton10} and
interception attacks~\cite{ballani:sigcomm07}. Similarly, there has been much
work on proposing secure routing protocols that can prevent the above
attacks~\cite{boldyreva:ccs12,gill:sigcomm11,hu:sigcomm04,oorschot:tissec07}. In
contrast, our work is the first one to show the
\emph{consequences} of these attacks on cryptocurrencies. 
\section{Conclusions}
\label{sec:conclusion}

This paper presented the first analysis of the vulnerabilities of the Bitcoin
system from a networking viewpoint. Based on real-world data, we showed that
Bitcoin is heavily centralized. Leveraging this fact, we then demonstrated and
quantified the disruptive impact that AS-level adversaries can have on the
currency. We showed that attackers can partition the Bitcoin network
by hijacking less than 100 prefixes. We also showed how AS-level adversaries
can significantly delay the propagation of blocks while remaining undetected.
Although these attacks are worrying, we also explained how to counter them with both short-term and long-term measures, some of which are easily deployable today.
\vfill
\section*{Acknowledgments}

We would like to thank Christian Decker for sharing Bitcoin data with us as
well as for his valuable comments in the beginning of the project. We would
also like to thank David Gugelmann for his support in one of our experiments.
Finally, we are grateful to Tobias B\"uhler, Edgar Costa Molero, and Thomas Holterbach from ETH
Z\"urich as well as Eleftherios Kokoris Kogias from EPFL for their helpful
feedback on early drafts of this paper. Aviv Zohar is supported by the Israel
Science Foundation (grant 616/13) and by a grant from the Hebrew University
Cybersecurity Center.
\clearpage
{
\bibliographystyle{IEEEtranS}
\bibliography{paper}
}

\clearpage
\appendix

\section{}
\subsection{Bitcoin event-driven simulation}
\label{ssec:simulation}
In this section, we provide more details on the simulator we used in Section~\ref{ssec:network_delay_attack}.

\myitem{Inputs} The simulator takes as input a realistic topology and some synthetic topologies with higher or lower degree of multi-homing. 
Each of the topologies includes the list of IP addresses running Bitcoin nodes, the IPs of the gateways and the hash rate
 associated with each mining pool, along with the forwarding paths among all pairs of IPs.
The realistic topology was inferred as described in Section~\ref{sec:topology}.
The synthetic ones were created by adding more gateways to the pools in the realistic topology 
 until they all reached the predefined degree of multi-homing.

\myitem{Model} The simulator models each Bitcoin node as an independent thread which reacts
to events according to the Bitcoin protocol. Whenever a node communicates with
another, the simulator adds a delay which is proportional to the number of ASes
present on the forwarding path between the two nodes. Each nodes initializes 8
connections following the default client implementation. Blocks are generated
at intervals drawn from an exponential distribution, with an expected rate of
one block every 10 minutes. The probability that a specific pool succeeds in
mining a block directly depends on its mining power. We consider that the
gateways of a pool form a clique and communicate with each other with links of
zero delay which an attacker cannot intercept. Concretely, this means that
whenever a block is produced by a pool, it is simultaneously propagated from
all the gateways of this pool. Although this choice makes the attacker less effective,
we assume that this is the default behavior of a benign pool.

\myitem{Attack} In the simulation, an AS can effectively delay the delivery of a block between
two nodes if and only if she intercepts the traffic from the potential
recipient to the sender, essentially if she is able to perform the attack
depicted in Fig.~\ref{fig:delay_attack_left}. An adversary that intercepts the
opposite direction is considered unable to delay blocks, during the simulation.
Although this choice makes the attacker less effective, it avoids a dependency
between the orphan rate and time. Indeed, these attackers lose their
power through time, as the connections they intercept are dropped after the
first effective block delay, and possibly replaced with connections that are
not intercepted by the attacker.

\myitem{Validation} We verified our simulator by comparing the orphan rate as well as
the median propagation delay computed with those of the real network. We found
that both of them are within the limits of the actual Bitcoin
network~\cite{decker2013information}.

\myitem{Methodology} We evaluated 
the impact of delay attacks considering the three most powerful
country-based coalition (US, China and Germany) as measured by the percentage
of traffic all their ASes\footnote{We map an AS to a country based on the country it is registered in.} see. 

\remove{We acknowledge that such 
coalitions are extremely unlikely to arise in practice and clearly constitute a
worst-case attacker for Bitcoin. Thus, we conclude that Bitcoin is not vulnerable to
network-wide delay attacks.
}
We run the simulation 20 times for each set of parameters and consider a new
random Bitcoin topology during each run. In each run, 144 blocks are created,
which is equivalent to a day's worth of block production (assuming an average
creation rate of one block per 10 minutes). We evaluated the impact of delay
attacks by measuring the orphan rate, different adversaries can cause.

\subsection{Possible Partitions}
Table~\ref{tab:partition_cost_appendix} shows a complete list of all the possible sets of nodes $P$ with the mining power that can be isolated. A summary of it was included and discussed in Section~\ref{sec:eval_partition}.
\label{ssec:partitions}

\begin{table}[t]
\centering
\small
\def\arraystretch{1.1}
\setlength{\tabcolsep}{11pt}
\begin{tabular}{c c c c}
\toprule
\begin{tabular}{@{}c@{}} \emph{Isolated} \\ \emph{mining power}\end{tabular}&\begin{tabular}{@{}c@{}} \emph{Minimum} \\ \emph{\# prefixes}\end{tabular}&\begin{tabular}{@{}c@{}} \emph{Median} \\ \emph{\# prefixes}\end{tabular}&\begin{tabular}{@{}c@{}} \emph{\# Feasible} \\ \emph{Partitions}\end{tabular}   \\

\midrule
6\%\ & 2 & 86 & 20\\
7\%\ & 7 & 72 & 23\\
8\%\ & 32 & 69 & 14\\
30\%\ & 83 & 83 & 1\\
39\%\ & 32 & 51 & 11\\
40\%\ & 37 & 80 & 8\\
41\%\ & 44 & 55 & 3\\
45\%\ & 34 & 41 & 5\\
46\%\ & 78 & 78 & 1\\
47\%\ & 39 & 39 & 1\\

\bottomrule
\end{tabular}
\caption{This table lists all partitions that can be created based on our inferred topology. The leftmost column indicates the portion of mining power contained within the isolated set and the rightmost the number of different combinations of pools that could form it.}
\label{tab:partition_cost_appendix}
\end{table}

\subsection{Inferring pool gateways}

\label{ssec:gw}
\remove{

The attacker begins with a list of IPs of suspected gateways used by the pool.
He first hijacks these IPs. Then, he releases some transaction $tx$ (with a
high fee) to the rest of the network. He additionally filters any \textsf{INV}
messages announcing this transactions to the pool's gateways. This can be done
by replacing the hash of the transaction with another when it is included in
the \textsf{INV} messages the gateways receives \maria{explain it somewhere else or
put it as node?} This implies that the pool cannot in fact learn about $tx$,
unless it has some other external connection that the attacker the existence of
which the attacker ignores. Once the pool creates its next block, the attacker
can check if $tx$ is included in this block. If it is, the set of known
gateways is not sufficient to isolate the pool.

Notice that a mining pool operator that does not monitor BGP messages will not
be aware that this probe is taking place, as it is none intrusive and does not
disrupt the regular operation of the pool. Once the pool constructed a block
without the transaction, the attacker can stop the hijack until he is ready to
perform his attack.

We note that if the list of gateways is incomplete, it is possible to further
act to identify unknown gateway nodes. Several techniques may reveal the
gateway.
}
We summarize here two ways hijacking can be used to reveal the
gateways of pools.

\begin{enumerate}
	\item The attacker may hijack the pool's stratum servers to discover connections that they establish and thus reveal the gateways they connect to. Connections between the pool's stratum server and its gateways are done via bitcoind's RPC access, and can be easily distinguished.
	\item The attacker may hijack the relay network to discover the connections used by large mining pools. The relay network has indeed a public set of six IPs that pools connect to. By hijacking these, one may learn the IPs of various gateways. Specifically, it will not be hard to identify a pool by observing the blocks each publishes to the relay network. 
\end{enumerate}

\subsection{Increase partition persistence}
\label{ssec:partition_life}
Our observations in Section~\ref{ssec:methodology_experiment} reveal concrete ways for the attacker to make her partition lasts longer. Intuitively, the attacker needs to suppress the effect of churn in order to keep the victim nodes isolated. The following three observations help in this regard: 

\newpage
 \emph{Observation 1: A smaller set of nodes will be easier to isolate for extended periods.} 
 The smaller the set of nodes is, the more time will pass before any isolated node restarts. Similarly, if the set of nodes is small,  external nodes will connect to the isolated set with lower probability. 
 
 \emph{Observation 2: All incoming connection slots can be occupied by connections from attacker nodes.}
 Bitcoin nodes typically limit the number of incoming connections they accept. The attacker may use several nodes to aggressively connect to the isolated nodes, and maintain these connections even after BGP routing is restored. Connections initiated by external nodes would then be rejected by the isolated nodes (as all slots of connections remain occupied).
 
 \emph{Observation 3: Outgoing connections can be biased via a traditional eclipse attack~\cite{heilman2015eclipse}.}
 Once an isolated node reboots, it will try to establish connections to nodes in its peer lists. These can be biased to contain attacker nodes, other isolated nodes, and IP addresses that do not actually lead to Bitcoin nodes. This is done by aggressively advertising these IPs to the victim nodes. The connections they form upon rebooting will then be biased towards those that maintain the partition.

\subsection{Frequently Asked Questions}
\label{ssec:faq}

\myitem{Can bitcoin relays bridge the partition?}
\label{sssec:relays}
While Bitcoin relays~\cite{relays} improve Bitcoin performance, they do not improve Bitcoin security against routing attacks.
Indeed, as the IPs of the relays are publicly available, relays are also vulnerable to both hijacking and passive attacks. As such, they can neither bridge the partition, nor act as a protected relay when an actual network attack takes place. Note that routing attacks are completely different from  DDoS attacks, which relays are most likely protected against. Routing attacks are much more complicated to deal with as they constitute a human-driven process depending highly on the provider of the victim rather than the victim itself. As an illustration, the mitigation of a hijack could include the provider of the attacker disconnecting her or her announcements being filtered globally, after the hijack has been detected by the provider of the victim.

\myitem{Can NATed nodes bridge a partition?}
 All nodes behind a NAT act as a simple node that doesn't accept connections. Indeed, by hijacking the corresponding public IP, the attacker receives traffic destined to all of them. In Section~\ref{sec:partition}, we explained why simple nodes cannot bridge the partition, the same applies for NATed nodes.

\myitem{Can Bitcoin routing alternatives solve the attacks presented in this paper?}  
Similar to the Bitcoin relays, existing proposals could indeed speed up the transmission of blocks by tackling with natural limitation of the current Internet such as delay and packet loss.
For instance, FIBRE~\cite{fibre} uses Forward Error Correction (FEC) and compression to speed up the block transmission. In parallel, Falcon~\cite{falcon}, uses a technique called cut-through routing to achieve fast block transmission in the presence of increased network delays. However, neither of the aforementioned approaches address routing attacks in the form of BGP hijacks or traffic eavesdropping.

\subsection{Ethical considerations}
\label{sssec:ethic}

Although we performed both of the attacks we describe against nodes that were connected to the Bitcoin network we did not disturb their normal operations.
Regarding the hijack experiment, we advertised and hijacked a prefix that was assigned to us by the Transit Portal (TP) project~\cite{peering}. No other traffic was influenced by our announcements. 
Also, the isolated nodes were experimental nodes we ran ourselves. Actual Bitcoin nodes that happened to be connected to these were not affected, they just had one of their connections occupied (as if they were connected to a supernode). 
Regarding the delay experiment, the victims were again our own Bitcoin clients. While they were indeed connected to nodes in the real network, those were not affected as we only modified packets towards our victim nodes.


\end{document}